\documentclass{IEEEtran}
\usepackage{float}
\usepackage{graphicx}
\usepackage{clrscode}
\usepackage{stfloats}
\usepackage{paralist}
\usepackage{tikz}
\usepackage{soul} 
\usepackage{color}
\usepackage{cite}
\usepackage{multicol}
\usepackage{comment}

\usepackage{amsmath, amssymb, mathrsfs, amsfonts}
\usepackage{amsthm}
\usepackage{mathtools}
\usepackage{epsfig, epstopdf}

\usepackage{algorithm}
\usepackage{enumitem}
\usepackage{algorithmic}
\usepackage{lipsum}
\usepackage{pbox}
\usepackage{array}
\usepackage{tablefootnote}

\usepackage{subfigure}

\usetikzlibrary{arrows,automata}

\allowdisplaybreaks

\newtheorem{lemma}{Lemma}

\usepackage{hyperref}

\newcommand{\RNum}[1]{\uppercase\expandafter{\romannumeral #1\relax}}

\begin{document}

\title{Secure Integrated
Sensing and Communication Networks: Stochastic Performance Analysis}

\author{
		\IEEEauthorblockN{Marziyeh Soltani, Mahtab Mirmohseni, \textit{Senior Member, IEEE}, and Rahim Tafazolli, \textit{Fellow, IEEE} \\
		\vspace*{0.5em}
			}\thanks{Part of this paper has been presented at the IEEE SPAWC 2025 \cite{OnStochasticPerformanceAnalysisofSecureIntegratedSensingandCommunicationNetworks}.\\
The authors are with 5/6GIC, the Institute for Communication Systems (ICS),
University of Surrey, GU2 7XH Guildford, U.K. (e-mail: m.soltani@surrey.ac.uk, m.mirmohseni@surrey.ac.uk, r.tafazolli@surrey.ac.uk).
This research is supported by the 5/6GIC, Institute for Communication
Systems (ICS), University of Surrey.}}
\maketitle
\begin{abstract}
This paper analyzes the stochastic security performance of a multiple-input multiple-output (MIMO) integrated sensing and communication (ISAC) system in a downlink scenario. A base station (BS) transmits a multi-functional signal to simultaneously communicate with a user, sense a target’s angular location, and counteract eavesdropping threats. The attack model considers a passive single-antenna  communication eavesdropper (eav) intercepting communication data, as well as a multi-antenna sensing eav attempting
to infer the target’s location. We also consider a malicious target scenario where the target plays the role of the communication eav. The BS-user and BS-eaves channels follow Rayleigh fading, while the target's azimuth angle is uniformly distributed. To evaluate the performance in this random network, we derive the ergodic secrecy rate (ESR) and the ergodic Cramér-Rao lower bound (CRB), for target localization, at both the BS and the sensing eav. This involves computing the probability density functions (PDFs) of the signal-to-noise ratio (SNR) and CRB, leveraging the central limit theorem for tractability. We characterize the boundary of the CRB-secrecy rate region, and interpret the performance tradeoffs between communication and sensing while guaranteeing a level of security and privacy in the random ISAC networks.
\end{abstract}
\section{Introduction}\label{introduction}
Recently, integrated sensing and communications (ISAC) has emerged as a promising paradigm for future wireless networks, enabling applications such as vehicular communication and smart city infrastructure \cite{Integratedtoward}. By sharing resources like spectrum and hardware, ISAC enhances both sensing and communication (S \& C). 

The shared spectrum, broadcast nature of wireless communication, and growing complexity of S\&C data exchanges pose major security risks for ISAC systems \cite{Sensing-AssistedEavesdropperEstimation:AnISACBreakthroughinPhysicalLayerSecurity}. Compared to traditional systems that only handle communication, ISAC systems face more complex security issues. These challenges mainly fall into two categories: \textit{keeping communication data secure} and \textit{preserving sensing security in radar-like sensing} \cite{SecuringtheSensingFunctionalityinISACNetworks}. The first challenge comes from the fact that ISAC systems use waveforms that carry information. These waveforms can be received by the targets being sensed, which might act as eavesdroppers (eaves) and access sensitive information. \textit{Sensing security} in ISAC is another concern, as adversaries can intercept reflected sensing signals to infer system activities, breaching targets privacy \cite{OnRadarPrivacyinSharedSpectrumScenarios}. Traditional methods like encryption often fail to meet ISAC’s low-latency needs and are prone to transmission parameter leakage \cite{IntegratingSensingandCommunicationsin6G}. Moreover, these methods can not preserve target privacy. This has sparked interest in physical layer security (PLS), which offers real-time protection without complex key management.

To enhance PLS in ISAC, \cite{PrivacyandSecurityinUbiquitousIntegrated} proposes strategies such as AI-based secure decision-making, friendly jamming, and RIS-assisted techniques. Existing ISAC security solutions that focus solely on either sensing privacy or communication security often rely on non-convex optimization techniques. These typically require iterative algorithms such as successive convex approximation (SCA) and semidefinite relaxation (SDR), which do not yield closed-form beamforming solutions \cite{OptimalBeamformingforsecureIntegratedSensingandCommunicationExploitingTargetLocation, SecuringtheSensingFunctionalityinISACNetworks}. Even in works like \cite{SecureCell-FreeIntegratedSensingandCommunicationinthePresenceofInformation}, which address transmit design for securing both sensing privacy and communication security, closed-form solutions remain elusive. Moreover, while random ISAC networks without security constraints have been explored in the literature \cite{cooperativeisacnetworksperformance, Network-levelISAC:AnAnalyticalStudyofAntenna, OnStochasticFundamentalLimitsinaDownlink}, a unified framework that jointly considers both communication security and sensing privacy in a random ISAC system has not yet been reported. We provide further discussion on the types of randomness considered in existing secure ISAC studies in Subsection \ref{comparison}. These research gaps form the primary motivation for our work.
\subsection{Our contributions}
To the best of our knowledge, this is the first study to investigate the joint security of communication and sensing in a random ISAC system. Specifically, we consider a downlink multiple-input multiple-output (MIMO) ISAC system where a multi-antenna base station (BS) simultaneously transmits confidential information to a user and estimates the angular position of a target. Meanwhile, two types of passive adversaries, an information eav and a sensing eav (which can be strong with knowledge of the BS’s transmitted signal or weak), attempts to intercept user data and extracts the target’s angular information, respectively. Moreover, we also consider the case where the communication eav is the target itself (i.e., a malicious target), rather than an external eav. The BS-user and BS-external eav channels are modeled as independent Rayleigh fading, while the angle of arrival (AoA) of the target at the BS and the sensing eav follows a uniform distribution\footnote{Our framework is flexible and can accommodate any target angle distribution.}. Given the use of large antenna arrays at the base station, a common setting in MIMO ISAC, we propose two ISAC beamforming schemes with closed-form expressions: SSJB (Secure Subspace Joint Beamforming) and SLB (Secure Linear Beamforming). We note that SSJB and SLB are based on two widely used beamforming techniques in MIMO ISAC, namely SJB and LB \cite{OnStochasticFundamentalLimitsinaDownlink}. The key idea behind SSJB and SLB precoding matrices is the incorporation of security into these conventional techniques—not in a blind manner, but in a principled way. Specifically, SSJB combines the sensing-optimal (without security) and secure communication-optimal precoders, while SLB merges the secure sensing-optimal and secure communication-optimal precoders, with the flexibility to include a communication-optimal design as well. To achieve this, we first derive the secure sensing-optimal beamforming solution for the estimation task. In fact, we analyze the performance of random ISAC systems using structured, suboptimal transmit beamformers, constructed as weighted combinations of precoders, each optimized for above subproblems. This approach, also adopted in massive MIMO ISAC studies such as \cite{PowerAllocationforMassiveMIMO-ISACSystems, MultipleTargetDetectioninCellFreeMassive}, provides analytical insights while avoiding the complexity and absence of closed-form solutions typically associated with non-convex optimization. 
The contributions of this paper are summarized as follows:
\\
$\bullet$ Formulation and derivation of the closed form optimal solution of a secure sensing problem, where the objective is to minimize the Cramér-Rao bounds (CRB) \footnote{CRB bounds the error of any unbiased estimator and is widely used to evaluate sensing performance} of the target’s AoA at the BS while maximizing the CRB at the sensing eav, subject to power constraints.
\\
$\bullet$ Proposing two precoding matrices, namely SSJB and SLB to: 1) Deliver confidential communication data to a single-antenna user, 2) Estimate the angular location of a target via echo signals, 3) Prevent sensing and communication eaves from intercepting both the user’s data and the target’s location, and 4) Prevent the target from intercepting the user's data.
\\
$\bullet$ Derivation of exact expressions for key metrics suitable for analyzing random ISAC networks under the two proposed precoding schemes: the ergodic secrecy rate (ESR) (including the ergodic rates of the user and both communication eaves, i.e., the malicious target or the external eav), the ergodic Cramér–Rao bound (CRB), denoted as $E[\text{CRB}]$, and the outage probability for target estimation, defined as $P(\text{CRB} > \epsilon)$, evaluated at both the BS and the sensing eav (for both weak and strong cases). We derive closed-form expressions for these metrics or provide upper, lower, and approximation bounds for them. In \cite{OnStochasticFundamentalLimitsinaDownlink}, the use of $P(\text{CRB} > \epsilon)$ is motivated as a meaningful sensing metric. Additionally, \cite{MIMOintegrated} demonstrates that $E[\text{CRB}]$ provides a tighter bound than the Bayesian CRB (BCRB), making it a key metric for evaluating ISAC security and privacy in random networks.
\\
$\bullet$ Validation of our framework through simulations and characterizing the CRB-secrecy rate region in random ISAC networks. Our results reveal three fundamental tradeoffs: 1) Sensing accuracy vs. communication rate as other traditional ISAC works with no security and privacy constraints; 2) Sensing accuracy vs. secure communication rate; 3) Communication security vs. sensing privacy.
These tradeoffs are consistent with open challenges noted in \cite{DelvingIntoSecurityandPrivacyofJoint}.
\subsection{Related works}
The fundamental limits of ISAC systems have been explored from various angles. Several studies analyze performance from an information-theoretic (IT) perspective, characterizing the capacity-distortion region \cite{2022AnInformation-TheoreticApproachtoJointSensingandCommunication}. Beyond IT, other works examine the S\&C trade-off using diverse sensing metrics and varying model randomness \cite{aframeworkformutualinformation,TowardDualfunctionalRadarCommunicationSystems,MUMIMOCommunicationsWithMIMORadar,JointTransmitBeamformingforMultiuser,optimaltransmitbeamformingintegrated,CrameRaoBoundOptimizationforJoint,fromtorchtoprojector,MIMOIntegratedSensingandCommunicationCRBRateTradeoff,fundamentalcrbratetradeoffmultiantenna,OnthePerformanceofUplinkandDownlink,PerformanceAnalysisandPowerAllocationforCooperative,NOMAISACPerformanceAnalysisandRateRegion,Aunifiedperformanceframeworkfor,MIMOISACPerformanceAnalysis,PerformanceAnalysioftheFullDuplexJoint,onthefundementaltradeoff,networklevelintegratedsensingcommunication,coverageandrateofjointcommunication}. For sensing metrics, \cite{aframeworkformutualinformation} uses estimation mutual information, while others rely on transmit beampattern \cite{TowardDualfunctionalRadarCommunicationSystems,MUMIMOCommunicationsWithMIMORadar,JointTransmitBeamformingforMultiuser,optimaltransmitbeamformingintegrated} and the CRB \cite{CrameRaoBoundOptimizationforJoint,fromtorchtoprojector,MIMOIntegratedSensingandCommunicationCRBRateTradeoff,fundamentalcrbratetradeoffmultiantenna}. However, the role of channel randomness, central to this paper, has largely been overlooked in these works. Recent studies have incorporated randomness, analyzing ISAC performance via probabilistic sensing metrics: detection probability \cite{OnthePerformanceofUplinkandDownlink,PerformanceAnalysisandPowerAllocationforCooperative,Aunifiedperformanceframeworkfor}, sensing rate/SNR \cite{NOMAISACPerformanceAnalysisandRateRegion,MIMOISACPerformanceAnalysis,PerformanceAnalysioftheFullDuplexJoint,networklevelintegratedsensingcommunication,coverageandrateofjointcommunication,PerformanceofDownlinkandUplinkIntegratedSensing}, and BCRB \cite{onthefundementaltradeoff}. Network-level ISAC has also been studied using stochastic geometry. Metrics include CRB and rate \cite{cooperativeisacnetworksperformance,Network-levelISAC:AnAnalyticalStudyofAntenna,ISACNetworkPlanning:SensingCoverageAnalysis,coverageandrateofjointcommunication}, communication coverage and radar information rate \cite{PerformanceAnalysisofCooperativeIntegratedSensing}, ISAC coverage probability\footnote{Defined as the weighted average of the probabilities that communication and sensing SIRs exceed their thresholds} \cite{OntheCoverageProbabilityinIntegratedSensingandCommunicationNetworkswithMulti-SOsInterference}, and ergodic radar/communication rates \cite{networklevelintegratedsensingcommunication}. These network level ISAC use zero-forcing (ZF) beamforming to mitigate the interference between S \& C, reduce precoding complexity, and make the
analysis tractable. In contrast, \cite{OnStochasticFundamentalLimitsinaDownlink} explores the full CRB–rate boundary assuming optimal beamforming correlated with both target and user channels, without nullifying S\&C impacts.

To address \textit{data communication security in ISAC systems}, \cite{SecureIntegratedSensingandCommunication} employs information-theoretic methods to explore inner and outer bounds of the secrecy-distortion region. In addition, various beamforming techniques have been developed to enhance PLS by degrading the SINR of eaves or maximizing secrecy capacity while preserving radar performance \cite{SecureDual-FunctionalRadar-CommunicationTransmission:ExploitingInterference,JointSecureTransmitBeamformingDesignsforIntegratedSensingandCommunicationSystems,SecurePrecodingOptimizationforNOMA-AidedIntegratedsensingandCommunication,JointPrecodingandArtificialNoiseDesignforSecure,JointBeamformingDesignforDual-FunctionalMIMORadarandCommunicationSystemsGuaranteeingPhysicalLayerSecurity,RobustTransmitBeamformingforSecureIntegratedSensingandCommunication,SecureRadar-CommunicationSystemsWithMaliciousTargets:IntegratingRadar,optimaltransmitbeamformingforsecrecy,PhysicalLayerSecurityOptimizationWithCramér–RaoBoundMetric,Sensing-AssistedEavesdropperEstimation:AnISACBreakthroughinPhysicalLayerSecurity,OptimalBeamformingforsecureIntegratedSensingandCommunicationExploitingTargetLocation,SecureISACMIMOSystems:ExploitingInterferenceWith,ANAssistedSecureISACBeamforming:CounteringCovertEavesdropper,RIS-basedPhysicalLayerSecurityforIntegratedSensingandCommunication:AComprehensiveSurvey,CovertBeamformingDesignforIntegratedRadarSensingandCommunication,RobustBeamformingDesignforCovertIntegratedSensingandCommunication,SensingforSecureCommunicationinISAC:ProtocolDesignand,BeamformingDesignforSecureMC-NOMAEmpowered}. Common radar metrics used in these designs include the SINR of radar echoes \cite{SecureDual-FunctionalRadar-CommunicationTransmission:ExploitingInterference,JointSecureTransmitBeamformingDesignsforIntegratedSensingandCommunicationSystems,SecurePrecodingOptimizationforNOMA-AidedIntegratedsensingandCommunication}, the mean square error (MSE) of beampattern design \cite{JointPrecodingandArtificialNoiseDesignforSecure,JointBeamformingDesignforDual-FunctionalMIMORadarandCommunicationSystemsGuaranteeingPhysicalLayerSecurity,RobustTransmitBeamformingforSecureIntegratedSensingandCommunication,SecureRadar-CommunicationSystemsWithMaliciousTargets:IntegratingRadar,optimaltransmitbeamformingforsecrecy}, and the CRB or BCRB for direction-of-arrival (DoA) estimation \cite{Sensing-AssistedEavesdropperEstimation:AnISACBreakthroughinPhysicalLayerSecurity,OptimalBeamformingforsecureIntegratedSensingandCommunicationExploitingTargetLocation,SecureISACMIMOSystems:ExploitingInterferenceWith,BeamformingDesignforSecureMC-NOMAEmpowered,PhysicalLayerSecurityOptimizationWithCramér–RaoBoundMetric}. In the context of covert communication, methods such as optimizing detection error probability, KL-divergence-based secrecy metrics \cite{CovertBeamformingDesignforIntegratedRadarSensingandCommunication,RobustBeamformingDesignforCovertIntegratedSensingandCommunication}, and AN-assisted secure null-space beamforming \cite{ANAssistedSecureISACBeamforming:CounteringCovertEavesdropper} have been explored to enhance stealth. \cite{SensingforSecureCommunicationinISAC:ProtocolDesignand} proposes a sensing-aided secure protocol, where the BS first estimates the eave’s CSI via beam sensing, then optimizes beamforming to reduce information leakage. For secure ISAC via intelligent reflecting surfaces (IRS), \cite{RIS-basedPhysicalLayerSecurityforIntegratedSensingandCommunication:AComprehensiveSurvey} provides a comprehensive survey, categorizing IRS-based PLS techniques into passive and active IRS paradigms.

Additionally, several studies address \textit{sensing security risks in ISAC} systems \cite{SecuringtheSensingFunctionalityinISACNetworks,IllegalSensingSuppressionforIntegratedSensingandCommunicationSystem,SecureCell-FreeIntegratedSensingandCommunicationinthePresenceofInformation,privacyperformanceofmimodulafunctional,MultiStaticISACinCellFree}. In a bi-static ISAC setting, \cite{SecuringtheSensingFunctionalityinISACNetworks} focuses on simultaneously supporting multiple user communications and bi-static target estimation with a legitimate radar receiver. The study formulates a mutual information (MI) maximization problem for legitimate sensing receivers while imposing MI constraints for sensing eaves and communication constraints for communication users, employing artificial noise to enhance privacy. Similarly, \cite{SecureCell-FreeIntegratedSensingandCommunicationinthePresenceofInformation} jointly optimizes S \& C beamforming for a cell-free ISAC systems, to maximize target detection probability, ensuring compliance with minimum SINR constraints for user equipment (UEs), maximum SINR constraints for information eaves, and maximum sensing probability thresholds for sensing eaves. In \cite{IllegalSensingSuppressionforIntegratedSensingandCommunicationSystem}, the target is itself a sensing eav. It focuses on protecting the privacy of legitimate users from adversaries equipped with sensing capabilities in ISAC systems. It examines user detection and DoA estimation at the sensing eav, aiming to impair the adversary’s sensing accuracy. To this end, the authors formulate a beamforming optimization problem designed to maximize both the miss-detection probability and the CRB. \cite{privacyperformanceofmimodulafunctional} investigates the
precoder design in a single ISAC transmitter scenario based on the
sensing beampattern distortion and introduces a sensing adversary estimation framework tailored for
estimating target location capitalizing on Bayesian inference. \cite{MultiStaticISACinCellFree} considers the precoder
design in a cell-free ISAC system based on the sensing SNR maximization. The study extends the sensing eav
model via exploiting an expectation maximization method to obtain an estimation of the transmitted signal, which is used to create replicas of the transmit beampattern of each transmitter and thus to eavesdrop target information.

\subsection{Comparison with existing work:}\label{comparison}
Among the works that consider only communication security risks and focus on estimation (i.e., those which have used CRB as the sensing metric similar to our work) \cite{Sensing-AssistedEavesdropperEstimation:AnISACBreakthroughinPhysicalLayerSecurity,OptimalBeamformingforsecureIntegratedSensingandCommunicationExploitingTargetLocation,SecureISACMIMOSystems:ExploitingInterferenceWith,PhysicalLayerSecurityOptimizationWithCramér–RaoBoundMetric} assume that the target itself is a malicious eav, with either a known channel, a known prior angle distribution, or a known angle estimate, obtained after the base station (BS) emits an omni-directional waveform and receives echoes reflected from both legitimate users and eaves located within the sensing range \cite{Sensing-AssistedEavesdropperEstimation:AnISACBreakthroughinPhysicalLayerSecurity}. In \cite{BeamformingDesignforSecureMC-NOMAEmpowered}, the target and the sensing eav are considered as different nodes. Among the studies that consider sensing privacy, \cite{SecuringtheSensingFunctionalityinISACNetworks,IllegalSensingSuppressionforIntegratedSensingandCommunicationSystem,privacyperformanceofmimodulafunctional,MultiStaticISACinCellFree} neglect the security of communication data, as they do not consider a communication eav or a malicious target. In \cite{SecureCell-FreeIntegratedSensingandCommunicationinthePresenceofInformation}, the communication eav is a separate node from the target. In contrast, our work jointly considers both sensing privacy and communication security, and accounts for both a malicious target (with known prior angle distribution) and external sensing and communication eaves (with unknown channels). Moreover, in \cite{Sensing-AssistedEavesdropperEstimation:AnISACBreakthroughinPhysicalLayerSecurity, IllegalSensingSuppressionforIntegratedSensingandCommunicationSystem, SecureCell-FreeIntegratedSensingandCommunicationinthePresenceofInformation}, the channels are assumed to be deterministic. In contrast, \cite{OptimalBeamformingforsecureIntegratedSensingandCommunicationExploitingTargetLocation, SecureISACMIMOSystems:ExploitingInterferenceWith, privacyperformanceofmimodulafunctional} model only the target’s location as a random variable, while \cite{PhysicalLayerSecurityOptimizationWithCramér–RaoBoundMetric, BeamformingDesignforSecureMC-NOMAEmpowered} consider the possible locations of the target and communication eav within a certain range. In contrast, our work models all channels (user, target, and eaves) as random parameters. Additionally, the sensing metrics used in \cite{SecuringtheSensingFunctionalityinISACNetworks} is mutual information (MI) at the BS and eav; \cite{privacyperformanceofmimodulafunctional} focuses on particle filter algorithm design; \cite{MultiStaticISACinCellFree} applies the expectation-maximization method to extract target information; and \cite{SecureCell-FreeIntegratedSensingandCommunicationinthePresenceofInformation} employs SNR and detection probability as sensing metrics. In contrast, we use $P(\text{CRB} > \epsilon)$ and $E[\text{CRB}]$ as our sensing performance metric, which provides a tighter bound than the BCRB or CRB used in other works \cite{Sensing-AssistedEavesdropperEstimation:AnISACBreakthroughinPhysicalLayerSecurity,OptimalBeamformingforsecureIntegratedSensingandCommunicationExploitingTargetLocation,SecureISACMIMOSystems:ExploitingInterferenceWith,BeamformingDesignforSecureMC-NOMAEmpowered,PhysicalLayerSecurityOptimizationWithCramér–RaoBoundMetric}. 

Finally, all the aforementioned works \cite{Sensing-AssistedEavesdropperEstimation:AnISACBreakthroughinPhysicalLayerSecurity,OptimalBeamformingforsecureIntegratedSensingandCommunicationExploitingTargetLocation,SecureISACMIMOSystems:ExploitingInterferenceWith,BeamformingDesignforSecureMC-NOMAEmpowered,SecuringtheSensingFunctionalityinISACNetworks,IllegalSensingSuppressionforIntegratedSensingandCommunicationSystem,privacyperformanceofmimodulafunctional,MultiStaticISACinCellFree,SecureCell-FreeIntegratedSensingandCommunicationinthePresenceofInformation} aim to find the optimal signaling strategy, resulting in non-convex problems and non-closed-form solutions. In contrast, our work proposes a closed-form beamforming solution, offering better analytical tractability. In our conference paper \cite{OnStochasticPerformanceAnalysisofSecureIntegratedSensingandCommunicationNetworks}, we proposed the SSJB scheme. In this extended work, we now introduce both SSJB and SLB, and significantly enhance the threat model by considering not only external communication eaves but also malicious targets. Furthermore, we account for both weak and strong sensing eaves, whereas the previous work considered only the latter.
\begin{figure}
    \includegraphics[scale=.63]{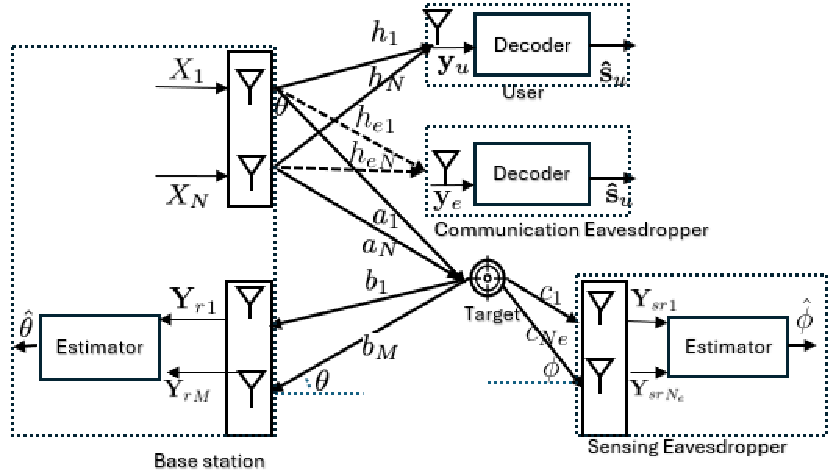}\caption{System model.${\mathbf{Y}_{sr}}_{i}$ and $h_{ei}$ denotes the $i$-th row and element of $\mathbf{Y}_{sr}$ and $\mathbf{h}_{e}$, respectively} \label{systemmodelfig}
\end{figure}
\textbf{Notation:} Scalars, vectors, and matrices are denoted by lowercase, boldface lowercase, and boldface uppercase letters, respectively. \( P(\cdot) \), \( f_x(\cdot) \), and \( E[\cdot] \) represent probability, PDF, and expectation. \( \mathbf{X}^{M \times N} \), \( \mathbf{X}^T \), \( \mathbf{X}^H \), and \( \mathbf{X}^* \) denote an \( M \times N \) matrix, its transpose, Hermitian transpose, and conjugate. The Euclidean norm is \( \|\cdot\| \), and the complex norm is \( |\cdot| \). \( \mathcal{CN}(\cdot,\cdot) \) and \( \mathcal{U}(a,b) \) denote circularly symmetric complex Gaussian and uniform distributions, respectively. \( \mathcal{N}_3(\boldsymbol{\mu}, \boldsymbol{\Sigma}) \) is a trivariate normal distribution with mean \( \boldsymbol{\mu} \) and covariance \( \boldsymbol{\Sigma} \). \( \mathbb{C} \) and \( \mathbb{R} \) denote complex and real number sets. \( \overset{d}{\rightarrow} \) and \( \overset{p}{\rightarrow} \) indicate convergence in distribution and probability. \( \text{Tr}(A) \), \( \otimes \), \( \text{vec} \), and \( \mathbf{I}_m \) represent the trace of \( A \), Kronecker product, vectorization, and the \( m \times m \) identity matrix. $F_{w,k,\lambda,s,m}(.)$ is the cumulative distribution function (CDF) of a generalized chi-square RV with parameters $w, k, \lambda, s, m$.
\section{System Model}\label{systemmodel}
We consider a BS with \(N\) transmit and \(M\) receive antennas serving a single-antenna communication user in the downlink. The BS knows the user’s channel and simultaneously senses a distant point target angle with an unknown location with a known distribution. Using a monostatic radar setup, the BS colocates its estimator and transmit antennas, ensuring identical angle of arrival (AoA) and departure (AoD). Two kinds of eaves are present:  
1) A single-antenna communication eav passively overhearing transmissions, with an unknown channel to the BS. We also consider the case where the target itself is the communication eav (the malicious target).
2) A sensing eav with \(N_e\) receive antennas and unknown channel to the BS, estimating the target’s position, threatening its privacy. Depending on its knowledge of the BS’s transmitted data, this eav can be either strong (with data knowledge) or weak (without it). The system model is shown in Fig. \ref{systemmodelfig}.

The channel vector between the BS and the user is given by \(\mathbf{h}=[h_1\quad h_2 \dots h_N]^T \in \mathbb{C}^{N \times 1}\), where the elements are independent and identically distributed (i.i.d.) according to \(\mathcal{CN} (0,1)\). More precisely, the \(i\)-th element can be expressed as \({h}_i = |{h}_i|e^{j{{\tilde{\phi}_i}}}\), where \(|{h}_i|\) follows a Rayleigh distribution with a scale parameter of \(1\), and \({{\tilde{\phi}_i}}\) is uniformly distributed over \([0, 2\pi)\). Furthermore, assuming an even number of antennas, the transmit and receive array steering vectors from the BS to the target are: $\mathbf{a}^{N \times 1}(\theta)  =  \![e^{-j\pi \sin(\theta)\frac{N-1}{2}}, e^{-j\pi \sin(\theta)\frac{N-3}{2}}, \dots, e^{j\pi \sin(\theta)\frac{N-1}{2}} ]^T$, and $\mathbf{b}^{M \times 1}(\theta) =  \![ e^{-j\pi \sin(\theta)\frac{M-1}{2}}, e^{-j\pi \sin(\theta)\frac{M-3}{2}}, \dots, e^{j\pi \sin(\theta)\frac{M-1}{2}}]^T,$ respectively, where \(\theta\) represents the azimuth angle of the target relative to the BS, which follows $\mathcal{U}(-\pi/2,\pi/2)$. The \(i\)-th element of \(\mathbf{a}(\theta)\) is expressed as \(a_i=e^{-jf_i}\), where \( f_i=\pi \sin(\theta)\frac{N-(2i-1)}{2} \). Moreover, the channel vector between the BS and the communication eav is given by \(\mathbf{h}_e \sim\mathcal{CN}(0,\mathbf{I}) \in \mathbb{C}^{N \times 1}\). The received steering vector at the sensing eav is given by: $\mathbf{c}^{N_e \times 1}(\phi) =[ e^{-j\pi \sin(\phi)\frac{N_e-1}{2}}, e^{-j\pi \sin(\phi)\frac{N_e-3}{2}}, \dots, e^{j\pi \sin(\phi)\frac{N_e-1}{2}}]^T$, 
where \(\phi\) represents the azimuth angle of the target relative to the sensing eav, which follows $\mathcal{U}(-\pi/2,\pi/2)$.
The BS transmits $\mathbf{X}$ with power $p_t$. The received signals at the user, the communication eav, and the target are given by \( \mathbf{y}_u = c_1\mathbf{h}^H \mathbf{X} + \mathbf{z}_u \), \( \mathbf{y}_{e} = c_2\mathbf{h}_e^H \mathbf{X} + \mathbf{z}_e \), and \( \mathbf{y}_t = c_5\mathbf{a}^H \mathbf{X} + \mathbf{z}_t \) respectively, where \( \mathbf{z}_u \), \( \mathbf{z}_e \in \mathbb{C}^{1 \times L} \), and \( \mathbf{z}_t \) are additive white Gaussian noise (AWGN) vectors, with each element following the distribution \( \mathcal{CN} (0,\sigma^2) \)\footnote{For simplicity, we assume the noise terms have the same variance and we ignore the link between the BS and the sensing eav.}. Here, \( c_1 \), \( c_2 \), and \( c_5 \) denote the complex-valued path gains (functions of distance) from the BS to the user, to the eav, and to the target, respectively. The reflected echo signal matrices at the BS and the sensing eav are $\mathbf{Y}_r = c_3 \mathbf{b}(\theta)\mathbf{a}(\theta)^H \mathbf{X} + \mathbf{Z}_r$ and $\mathbf{Y}_{sr} = c_4 \mathbf{c}(\phi)\mathbf{a}(\theta)^H \mathbf{X} + \mathbf{Z}_{sr},$
respectively, where \( c_3, c_4 \in \mathbb{C} \) are the complex-valued channel coefficients, which depend on the target’s radar cross-section (RCS) and the round-trip path loss between the target-BS (which is $|c_5|^2$) and target-sensing eav, respectively. \( \mathbf{Z}_r \in \mathbb{C}^{M \times L} \) and \( \mathbf{Z}_{sr} \in \mathbb{C}^{N_e \times L} \) are AWGN matrices, with elements that are i.i.d. and follow the distribution \( \mathcal{CN} (0,\sigma^2_r) \).
\subsection{Metrics}\label{metric}
The metric used to evaluate sensing performance at the BS is $\text{CRB}(\theta)$. To assess sensing privacy, we use $\text{CRB}(\phi)$, which reflects how accurately the sensing eav can estimate the target's angle relative to itself. Secure communication performance is evaluated using the secrecy rate, which involves computing the user’s achievable rate $R$, the information leakage rate to the communication eav $R_e$, and, in the case of a malicious target, the leakage rate $R_t$.

The transmitted ISAC signal $\mathbf{X}$ is designed based on the channel state information and instantaneous realizations of the random parameters in the network. We propose two precoding strategies: SSJB and SLB. The details of these precoding designs are provided in Section \ref{isactransmission}. Once designed, $\mathbf{X}$ becomes a function of random variables, making all performance metrics, i.e., $\text{CRB}(\theta)$, $\text{CRB}(\phi)$, $R$, $R_e$, and $R_t$, random variables as well, since they depend on both $\mathbf{X}$ and the underlying channel realizations. Therefore, to evaluate the long-term performance of the network, we require appropriate metrics that account for randomness. In other words, the impact of randomness is reflected in the performance evaluation rather than in the design of $\mathbf{X}$. This assumption is commonly adopted in the literature \cite{MIMOIntegratedSensingandCommunicationCRBRateTradeoff, CrameRaoBoundOptimizationforJoint, PhysicalLayerSecurityOptimizationWithCramér–RaoBoundMetric}, which considers a radar target-tracking stage where the target moves slowly enough for the BS to obtain coarse estimates of parameters from a previous stage \cite{CrameRaoBoundOptimizationforJoint, MIMOIntegratedSensingandCommunicationCRBRateTradeoff}. Based on these estimates, the BS designs $\mathbf{X}$, treating it as fixed variable in designing.

To address the need for suitable performance metrics, we adopt $P(\text{CRB}(\theta) > \epsilon)$ and $P(\text{CRB}(\phi) > \epsilon)$ as sensing metrics, as proposed in \cite{OnStochasticFundamentalLimitsinaDownlink}, which are well-suited for real-time sensing systems requiring a high probability that estimation errors remain below a given threshold. Additionally, we use the ergodic CRB as a complementary performance metric, defined as \( E[\text{CRB}(\theta)] \) at the BS, and \( E[\text{CRB}(\phi)] \) at the sensing eav. In \cite{MIMOintegrated,OnStochasticFundamentalLimitsinaDownlink}, it is shown that ergodic CRB serves as a lower bound for the estimation error and provides a tighter bound compared to Bayesian CRBs.
Moreover, the achievable ESR for the user is given by \cite{SecureTransmissionWithArtificialNoiseOverFadingChannels}: $
C_s = \left( E[R] - E[R_e] \right)^+ = \left( \mathbb{E}_{\mathbf{h}, \mathbf{a}(\theta)} \left[ \log(1 + \text{SINR}_u) \right] - \mathbb{E}_{\mathbf{h}, \mathbf{a}(\theta), \mathbf{h}_e} \left[ \log(1 + \text{SINR}_e) \right] \right)^+,
$ where $\text{SINR}_u$ and $\text{SINR}_e$ denote the SINR at the user and external eav, respectively. If the target itself is malicious and acts as a communication eav, the ESR becomes: $
C_s = \left( E[R] - E[R_t] \right)^+ = \left( \mathbb{E}_{\mathbf{h}, \mathbf{a}(\theta)} \left[ \log(1 + \text{SINR}_u) \right] - \mathbb{E}_{\mathbf{h}, \mathbf{a}(\theta)} \left[ \log(1 + \text{SINR}_t) \right] \right)^+,
$ where $\text{SINR}_t$ is the SINR at the malicious target.
\section{ CRB analysis}\label{crbanalysis}
In this section, we derive $\text{CRB}(\theta)$ and $\text{CRB}(\phi)$, which form the foundation for designing the precoding matrix $\mathbf{X}$.
\subsection{$\text{CRB}(\theta)$ at the BS}\label{opoftargetlb}
To derive the CRB of $\theta$, which is our sensing parameter of interest to be estimated at the BS, we use the received echo signal at the BS which is $\mathbf{Y}_{r} = c_3 \mathbf{b}(\theta)\mathbf{a}(\theta)^H \mathbf{X} + \mathbf{Z}_{r},$ 
and we obtain the Fisher information matrix (FIM) for estimating \( \xi=[\theta,\mathcal{R}(c_3),\mathcal{I}(c_3)]^T \in \mathbb{R}^{3 \times 1} \). We note that the overall reflection coefficient $c_3$ is an unknown deterministic parameter that must also be estimated to accurately determine the random parameter $\theta$. Let \( \mathbf{A}(\theta) =\mathbf{b}(\theta) \mathbf{a}(\theta)^T \), the received echo signal at the BS can be rewritten as $\mathbf{Y}_{r} = c_3 \mathbf{A}(\theta) \mathbf{X} + \mathbf{Z}_{r}$. By vectorizing $\mathbf{Y}_{r}$, we have: $\tilde{\mathbf{y}}_{r} = \text{vec}(\mathbf{Y}_{r}) = \tilde{\mathbf{u}} + \tilde{\mathbf{n}},$ where \( \tilde{\mathbf{u}} = c_3 \text{vec}(\mathbf{A}\mathbf{X}) \) and \( \tilde{\mathbf{Z}_{r}} = \text{vec}(\mathbf{Z}_{r}) \sim \mathcal{CN}(0, \sigma^2_r\mathbf{I}_{ML}) \). Thus, we have: $\tilde{\mathbf{y}}_{r}\sim \mathcal{CN}(\tilde{\mathbf{u}}, \sigma^2_r\mathbf{I}_{ML})$. Let \( {\mathbf{F}} \in \mathbb{R}^{3 \times 3} \) denote the FIM for estimating $\xi$ based on $\tilde{\mathbf{y}}_{r}$. Each element of \( {\mathbf{F}} \) is given by \cite{Fundamentalsofstatisticalsignalprocessing}:
\begin{equation}
{\mathbf{F}}_{i,j} = \text{tr} \left\{ \mathbf{R}^{-1} \frac{\partial \mathbf{R}}{\partial \xi_i} \mathbf{R}^{-1} \frac{\partial \mathbf{R}}{\partial \xi_j} \right\}
+ 2 \Re \left\{ \frac{\partial \tilde{\mathbf{u}}^H}{\partial \xi_i} \mathbf{R}^{-1} \frac{\partial \tilde{\mathbf{u}}}{\partial \xi_j} \right\},\label{elemntoffisher}
\end{equation}
for $i,j \in \{1,2,3\}$ where $\mathbf{R}$ is the covariance matrix of the Gaussian observation which in our case is $\sigma^2_r\mathbf{I}_{ML}$. By defining $\bar{\alpha}=[\mathcal{R}(c_3),\mathcal{I}(c_3)]^T \in \mathbb{R}^{2 \times 1}$, the FIM \( {\mathbf{F}} \) is partitioned as: ${\mathbf{F}} = 
\begin{bmatrix}
{\mathbf{F}}_{\theta\theta} & {\mathbf{F}}_{\theta\bar{\alpha}} \\
{\mathbf{F}}_{\bar{\alpha}{\theta}} & {\mathbf{F}}_{\bar{\alpha}\bar{\alpha}}
\end{bmatrix}.$
The covariance matrix \( \mathbf{R} \) is independent of \( \xi \). Therefore, we have \( \frac{\partial \mathbf{R}}{\partial \xi_i} = 0 \), \( i = 1,2,3 \) and the first term in (\ref{elemntoffisher}) is zero. Furthermore, $\frac{\partial \tilde{\mathbf{u}}}{\partial \theta} = c_3 \text{vec}(\mathbf{\dot{A}} \mathbf{X})+c_3 \text{vec}(\mathbf{A} \mathbf{\dot{X}})$ and $\frac{\partial \tilde{\mathbf{u}}}{\partial \bar{\alpha}} = [1, j] \otimes \text{vec}(\mathbf{A} \mathbf{X})$. Accordingly, the elements of ${\mathbf{F}}$ are given in (\ref{elemntoffisher2}), where $\mathbf{R}_x=\frac{1}{L}\mathbf {X}\mathbf{X}^H, $ \( j = \sqrt{-1} \), and \( \dot{\mathbf{A}} = \frac{\partial \mathbf{A}}{\partial \theta} \) and \( \dot{\mathbf{X}} = \frac{\partial \mathbf{X}}{\partial \theta} \) denote the partial derivative of \( \mathbf{A} \) and $\mathbf{X}$ w.r.t. \( \theta \), respectively. In (\ref{elemntoffisher2}), (a) follow from the identities \((\text{vec}(A))^H \text{vec}(B) = \text{trace}(A^H B)\) and \(\text{tr}(ABC) = \text{tr}(CAB)\). Next, we derive the CRB for estimating $\theta$, which corresponds to the first diagonal element of \( {\mathbf{F}}^{-1} \), i.e.,
\begin{equation}
\text{CRB}(\theta) = [{\mathbf{F}}^{-1}]_{1,1} = \left[{\mathbf{F}}_{\theta\theta} - {\mathbf{F}}_{\theta\tilde{\alpha}}{\mathbf{F}}_{\tilde{\alpha}\tilde{\alpha}}^{-1} {\mathbf{F}}_{\tilde{\alpha}\theta} \right]^{-1},\label{crbphi}
\end{equation}
where ${\mathbf{F}}_{\theta\tilde{\alpha}}=({\mathbf{F}}_{\tilde{\alpha}\theta})^T$ since the FIM is a symmetric matrix. Moreover, based on the design criteria outlined in Subsection \ref{metric}, we set $\dot{\mathbf{X}} = \frac{\partial \mathbf{X}}{\partial \theta} = 0$ at (\ref{elemntoffisher2}). This allows us to optimize the derived $\text{CRB}(\theta)$ with respect to $\mathbf{X}$, after which $\mathbf{X}$ becomes a function of the random variables. It is important to note that we do not use the Bayesian Cramér-Rao Bound (BCRB), which requires accounting for the distribution of the random parameter when deriving the elements of the FIM \footnote{At BCRB, we have: $\text{CRB}(\theta) = [{E[\mathbf{F}}_{\theta\theta}] + E[\frac{\partial \ln p_\theta(\theta)}{\partial \theta} \frac{\partial \ln p_\theta(\theta)}{\partial \theta^\text{T}} ]$-
$E[{\mathbf{F}}_{\theta\tilde{\alpha}}](E[{\mathbf{F}}_{\tilde{\alpha}\tilde{\alpha}}])^{-1} E[{\mathbf{F}}_{\tilde{\alpha}\theta}]]^{-1}$.}. Instead, we treat the CRB—originally derived for a fixed parameter—as a function of the random variables $\mathbf{h}$, $\theta$, and $\phi$. When the distributions of these variables are known, as in our case, we can evaluate the expected value $E(\text{CRB})$ and the probability $P(\text{CRB})$ to analyze the behavior of the CRB. This approach, which treats the CRB as a random variable, is employed in \cite{cramerraoboundonaerospaceandelectronicsystems, recentinsightsinthebayesiananddeterministic}.
\begin{figure*}[t]
\normalsize
\begin{align}
{\mathbf{F}}_{\theta\theta} &= \frac{2|c_3|^2}{\sigma_R^2} \Re \left\{ ( \text{vec}(\mathbf{\dot{A}} \mathbf{X}))^H \text{vec}(\mathbf{\dot{A}} \mathbf{X})+( \text{vec}(\mathbf{\dot{A}} \mathbf{X}))^H \text{vec}(\mathbf{A} \mathbf{\dot{X}})+( \text{vec}(\mathbf{A} \mathbf{\dot{X}}))^H \text{vec}(\mathbf{\dot{A}} \mathbf{X})+( \text{vec}(\mathbf{A} \mathbf{\dot{X}}))^H \text{vec}(\mathbf{A} \mathbf{\dot{X}}) \right\}\nonumber\\
&\overset{a}{=}\frac{2 |c_3|^2}{\sigma_R^2}\left\{ L\text{tr}(\mathbf{\dot{A}} \mathbf{R}_x \mathbf{\dot{A}}^H)+\text{tr}(\mathbf{{A}} \mathbf{\dot{X}} \mathbf{{X}}^H\mathbf{\dot{A}}^H)+\text{tr}(\mathbf{{\dot{A}}} \mathbf{{X}} \mathbf{{\dot{X}}}^H\mathbf{{A}}^H)+\text{tr}(\mathbf{{A}} \mathbf{\dot{X}} \mathbf{{\dot{X}}}^H\mathbf{{A}}^H),\right\}\nonumber\\
{\mathbf{F}}_{\theta\bar{\alpha}} &= \frac{2}{\sigma_R^2} \Re \left\{ (c_3^* \text{vec}(\mathbf{\dot{A}} \mathbf{X})^H+c_3^* \text{vec}(\mathbf{{A}} \mathbf{\dot{X}})^H) [1, j] \otimes \text{vec}(\mathbf{A} \mathbf{X}) \right\}\overset{a}{=}\frac{2}{\sigma_R^2} \Re \left\{ Lc_3^* (\text{tr}(\mathbf{A} \mathbf{R}_x \mathbf{\dot{A}}^H)+\text{tr}(\mathbf{A} \mathbf{X}\mathbf{\dot{X}}^H \mathbf{A}^H) )[1, j] \right\},\nonumber\\
\tilde{\mathbf{F}}_{\bar{\alpha}\bar{\alpha}} &= \frac{2}{\sigma_R^2} \Re \left\{ ([1, j] \otimes \text{vec}(\mathbf{A} \mathbf{X}))^H ([1, j] \otimes \text{vec}(\mathbf{A} \mathbf{X})) \right\}\overset{a}{=}\frac{2}{\sigma_R^2} \Re \left\{ ([1, j]^H [1, j]) (\text{tr}(\mathbf{A} \mathbf{X})^H \mathbf{A} \mathbf{X}) \right\}\label{elemntoffisher2}.\\
\text{CRB}(\theta)&=\frac{\sigma^2_R \text{Tr}(\mathbf{A}^H(\theta) \mathbf{A}(\theta) \mathbf{R}_x)}{2 \mid c_3 \mid ^2 L (\text{Tr}(\mathbf{A}^H(\theta) \mathbf{A}(\theta) \mathbf{R}_x) \text{Tr}(\dot{\mathbf{A}}^{H}(\theta) \dot{\mathbf{A}}^(\theta) \mathbf{R}_x)-\mid \text{Tr}(\dot{\mathbf{A}}^{H}(\theta) \mathbf{A}(\theta) \mathbf{R}_x)\mid^2)}. \label{crb}
\end{align}
\hrulefill
\end{figure*}
Thus, we obtain: ${\mathbf{F}}_{\theta\theta} =\frac{2L |c_3|^2}{\sigma_R^2} \{\text{tr}(\mathbf{\dot{A}} \mathbf{R}_x \mathbf{\dot{A}}^H)\}$, ${\mathbf{F}}_{\theta\bar{\alpha}}=\frac{2L}{\sigma_R^2} \Re \left\{ c_3^* (\text{tr}(\mathbf{A} \mathbf{R}_x \mathbf{\dot{A}}^H))[1, j] \right\}$, and $\tilde{\mathbf{F}}_{\bar{\alpha}\bar{\alpha}}=\frac{2L}{\sigma_R^2} \text{tr}(\mathbf{A} \mathbf{R}_x \mathbf{A}^H) \mathbf{I}_2$. After some mathematical operation and using the FIM elements, (\ref{crbphi}) simplifies to (\ref{crb}). Moreover, we have: 
\begin{align}
&\text{Tr}(\mathbf{A}^H(\theta) \mathbf{A}(\theta) \mathbf{R}_x)=\text{Tr}(\mathbf{a}\mathbf{b}^H\mathbf{b}\mathbf{a}^H \mathbf{R}_x)\overset{(a)}=||\mathbf{b}||^2\mathbf{a}^H\mathbf{R}_x\mathbf{a}
\nonumber\\
&\text{Tr}(\mathbf{A}'^H(\theta) \mathbf{A}'(\theta) \mathbf{R}_x)=\text{Tr}((\mathbf{a}\mathbf{b}'^H+\mathbf{a}'\mathbf{b}^H)(\mathbf{b}'\mathbf{a}^H+\mathbf{b}\mathbf{a}'^H) \mathbf{R}_x)\nonumber\\
&\overset{(a)}=(||\mathbf{b}'||^2\mathbf{a}^H\mathbf{R}_x\mathbf{a}+||\mathbf{b}||^2\mathbf{a}'^H\mathbf{R}_x\mathbf{a}') 
\nonumber\\
&\text{Tr}(\mathbf{A}'^H(\theta) \mathbf{A}(\theta) \mathbf{R}_x)\overset{(a)}=||\mathbf{b}||^2\mathbf{a}^H\mathbf{R}_X\mathbf{a}',
\end{align}
where (a) are due to $\mathbf{a}^H\mathbf{a}' = \mathbf{a}'^H\mathbf{a} = \mathbf{b}^H\mathbf{b}' = \mathbf{b}'^H\mathbf{b} = 0$ and the property $\text{Tr}(abc) = \text{Tr}(bca) = \text{Tr}(cab)$. Thus, by defining $Q \triangleq \frac{\sigma^2_R}{2 |c_3|^2 L},$ we obtain:
\begin{align}
\text{CRB}(\theta)=\frac{Q}{||\mathbf{b}'||^2\mathbf{a}^H\mathbf{R}_x\mathbf{a}+||\mathbf{b}||^2\mathbf{a}'^H\mathbf{R}_x\mathbf{a}'-\frac{|||\mathbf{b}||^2\mathbf{a}^H\mathbf{R}_X\mathbf{a}'|^2}{||\mathbf{b}||^2\mathbf{a}^H\mathbf{R}_x\mathbf{a}}}.\label{crbnew1}
\end{align}
\subsection{$\text{CRB}(\phi)$ at the sensing eavesdropper}\label{opoftargeteav}
In the received echo signal at the sensing eav, $\mathbf{Y}_{sr} = c_4 \mathbf{c}(\phi)\mathbf{a}(\theta)^H \mathbf{X} + \mathbf{Z}_{sr}$, we assume that \( c_4 \) is also an unknown but deterministic parameter. However, our primary parameter of interest is \( \phi \), which is to be estimated. In the following, we obtain the FIM for estimating \( \xi=[\phi,\mathcal{R}(c_4),\mathcal{I}(c_4)]^T \in \mathbb{R}^{3 \times 1} \) to facilitate the derivation of the $\text{CRB}(\phi)$. Let \( \mathbf{B}(\theta,\phi) =\mathbf{c}(\phi) \mathbf{a}(\theta)^T \), the received echo signal at the sensing eav can be rewritten as $\mathbf{Y}_{sr} = c_4 \mathbf{B}(\theta,\phi) \mathbf{X} + \mathbf{Z}_{sr}.$ For notational convenience, in the sequel we drop \( \theta \) and $\phi$ in \( \mathbf{B}(\theta,\phi) \), \( \mathbf{a}(\theta) \), and \( \mathbf{c}(\phi) \). By vectorizing $\mathbf{Y}_{sr}$, we have: $\tilde{\mathbf{y}}_{sr} = \text{vec}(\mathbf{Y}_{sr}) = \tilde{\mathbf{u}}_e + \tilde{\mathbf{n}},$ where \( \tilde{\mathbf{u}}_e = c_4 \text{vec}(\mathbf{B}\mathbf{X}) \) and \( \tilde{\mathbf{Z}_{sr}} = \text{vec}(\mathbf{Z}_{sr}) \sim \mathcal{CN}(0, \sigma^2_r\mathbf{I}_{N_eL}) \). By assuming that the eav is strong and it know $\mathbf{X}$ in the following we derive FIM. We have: $\tilde{\mathbf{y}}_{sr}\sim \mathcal{CN}(\tilde{\mathbf{u}}_e, \sigma^2_r\mathbf{I}_{N_eL})$. By defining $\bar{\alpha}_e=[\mathcal{R}(c_4),\mathcal{I}(c_4)]^T \in \mathbb{R}^{2 \times 1}$, we have $\frac{\partial \tilde{\mathbf{u}}_e}{\partial \phi} = c_4 \text{vec}(\mathbf{\dot{B}} \mathbf{X})+c_4 \text{vec}(\mathbf{B} \mathbf{\dot{X}})$ and $\frac{\partial \tilde{\mathbf{u}}_e}{\partial \bar{\alpha}_e} = [1, j] \otimes \text{vec}(\mathbf{B} \mathbf{X})$. Accordingly, the element of ${\mathbf{F}}$ is obtained as (\ref{elemntoffisher2}) by substituting $\mathbf{A}$ with $\mathbf{B}$, $\bar{\alpha}$ with $\bar{\alpha}_e$, $\theta$ with $\phi$, and $c_3$ with $c_4$, respectively. Also, \( \dot{\mathbf{B}} = \frac{\partial \mathbf{B}}{\partial \phi} \) and \( \dot{\mathbf{X}} = \frac{\partial \mathbf{X}}{\partial \phi}=0 \) denoting the partial derivative of \( \mathbf{B} \) and $\mathbf{X}$ w.r.t. \( \phi \), respectively. Thus, (\ref{elemntoffisher2}) is simplified to
${\mathbf{F}}_{\phi\phi}=\frac{2L |c_4|^2}{\sigma_R^2} \{\text{tr}(\mathbf{\dot{B}} \mathbf{R}_x \mathbf{\dot{B}}^H)\}$, ${\mathbf{F}}_{\phi\bar{\alpha}_e} =\frac{2L}{\sigma_R^2} \Re \left\{ c_4^* (\text{tr}(\mathbf{B} \mathbf{R}_x \mathbf{\dot{B}}^H))[1, j] \right\}$, and ${\mathbf{F}}_{\bar{\alpha}_e\bar{\alpha}_e}=\frac{2L}{\sigma_R^2} \text{tr}(\mathbf{B} \mathbf{R}_x \mathbf{B}^H) \mathbf{I}_2$. Thus, similar to (\ref{crbphi}), we have: $\text{CRB}(\phi) = \left[{\mathbf{F}}_{\phi\phi} - {\mathbf{F}}_{\phi\bar{\alpha}_e}{\mathbf{F}}_{\bar{\alpha}_e\bar{\alpha}_e}^{-1} {\mathbf{F}}_{\bar{\alpha}_e\phi} \right]^{-1}$, where ${\mathbf{F}}_{\phi\bar{\alpha}_e}=({\mathbf{F}}_{\bar{\alpha}_e\phi})^T$. Thus, after some mathematical manipulation we have: $\text{CRB}(\phi) = \frac{\sigma_R^2}{2L|c_4|^2 \left( \text{tr}(\mathbf{\dot{B}} \mathbf{R}_x \mathbf{\dot{B}}^H) - \frac{|\text{tr}(\mathbf{B} \mathbf{R}_x \dot{\mathbf{B}}^H)|^2}{\text{tr}(\mathbf{B} \mathbf{R}_x \mathbf{B}^H)} \right)}$, where $\mathbf{\dot{B}}=\mathbf{\dot{c}}(\phi)\mathbf{a}^H(\theta)$. 
Moreover, we have: $\text{tr}(\mathbf{B} \mathbf{R}_x \dot{\mathbf{B}}^H)=\text{Tr}(\mathbf{c}\mathbf{a}^H\mathbf{R}_x\mathbf{a}(\theta)\mathbf{\dot{c}}^H(\phi))\overset{a}{=}0$ and $\text{tr}(\mathbf{\dot{B}} \mathbf{R}_x \mathbf{\dot{B}}^H)=\text{Tr}(\mathbf{\dot{c}}(\phi)\mathbf{a}^H(\theta) \mathbf{R}_x\mathbf{a}(\theta)\mathbf{\dot{c}}^H(\phi))\overset{(a)}=||\mathbf{c}'||^2\mathbf{a}^H\mathbf{R}_x\mathbf{a},$ where (a)s are due to $\text{Tr}(abc)=\text{Tr}(bca)=\text{Tr}(cab)$, and \(\mathbf{c'}^H(\phi)\mathbf{c}(\phi) = 0\). Therefore, we have: $(\text{CRB}(\phi) = \frac{\sigma^2_R}{2 |c_4|^2 L ||\mathbf{c}'||^2\mathbf{a}^H\mathbf{R}_x\mathbf{a},}$ where $||\mathbf{c}'||^2 = \frac{\pi^2 \cos^2(\phi) N_e (N_e^2 - 1)}{12}$.
\section{ISAC Transmission Strategy}\label{isactransmission}
Optimizing beamforming for secure or private ISAC lacks a closed-form solution, as shown in \cite{Sensing-AssistedEavesdropperEstimation:AnISACBreakthroughinPhysicalLayerSecurity,OptimalBeamformingforsecureIntegratedSensingandCommunicationExploitingTargetLocation,SecureISACMIMOSystems:ExploitingInterferenceWith,BeamformingDesignforSecureMC-NOMAEmpowered,SecuringtheSensingFunctionalityinISACNetworks,IllegalSensingSuppressionforIntegratedSensingandCommunicationSystem,privacyperformanceofmimodulafunctional,MultiStaticISACinCellFree,SecureCell-FreeIntegratedSensingandCommunicationinthePresenceofInformation}. To exploit the large antenna arrays at the BS, typical in MIMO ISAC, we propose two structured, suboptimal transmit beamformers, constructed as weighted combinations of precoders, each optimal for a specific subproblem. This approach, inspired by prior massive MIMO ISAC works \cite{PowerAllocationforMassiveMIMO-ISACSystems, MultipleTargetDetectioninCellFreeMassive}, offers analytical tractability without the complexity of non-convex optimization. In the following, we elabotare these subproblems.
\\
\textbf{Subproblem 1 (sensing optimal without security):} The optimal beamforming vector that minimizes the CRB at the BS for target estimation while maximizing the communication rate, subject to a power constraint—that is, the optimal solution to:
\begin{align}
&\underset{\mathbf{w}}{\min} \quad \mathrm{CRB}(\theta) \nonumber\\
& \text{s.t.} \quad \gamma \leq \log_{2}\Big(1 + \frac{|\mathbf{h}^H \mathbf{w}|^2}{\sigma^2_u}\Big), \nonumber\\
& \mathrm{tr}(p_t \mathbf{w} \mathbf{w}^H) \leq p_t,
\end{align}
lies in the span of \(\{\mathbf{a}, \mathbf{h}\}\), as shown in \cite[Lemma 1]{CrameRaoBoundOptimizationforJoint}. 
\\
\textbf{Subproblem 2 (secure communication optimal):} When the communication eav’s CSI is unavailable, secrecy performance metrics include the ESR and secrecy outage-based secrecy rate \cite{Artificial-Noise-AidedBeamformingDesignintheMISOME}. Regarding ESR, a near-optimal strategy is the masked beamforming method \cite{SecureTransmissionWithMultipleAntennasI}, where the beamforming vector is based only on the legitimate user’s channel. In this method, the transmitter sends the message (encoded with a scalar Gaussian wiretap code) along the user’s channel to maximize the signal power. Also, it injects spatio-temporal white noise into the null space of the user’s channel, avoiding interference while degrading the eav’s reception. Regarding secrecy outage, artificial noise (AN) distributed uniformly in the null space of the information beam is proven to be optimal \cite{SecrecyOutageinMISOSystemsWithPartialChannelinformation}. 
\\
\textbf{Subproblem 3 (secure sensing optimal):} In a system where the BS operates as a monostatic MIMO radar to estimate a target’s location, while a sensing eav attempts to infer the target’s angle, the optimal transmit beamformer, derived in Lemma \ref{privacy} (proof in Appendix \ref{privacyproof}), aims to minimize $\text{CRB}(\theta)$ at the BS and maximize $\text{CRB}(\phi)$ at the eav \footnote{If the CRB is large, it means no unbiased estimator can have low variance,
and estimation is inherently less precise.}, thereby enhancing estimation accuracy at the BS while degrading it at the eav.
\begin{lemma}\label{privacy}
The optimal solution of:
\begin{align}
&\underset{\mathbf{R}_x}{\min} \quad \text{CRB}(\theta)- \text{CRB}(\phi)\nonumber\\
& \text{s.t.} \quad \mathbf{R}_x \geq 0; \nonumber\\
& \mathrm{tr}( \mathbf{R}_x)= p_t, \label{optimization}
\end{align}
lies in the span of $\frac{\mathbf{a}}{\|\mathbf{a}\|}$ and $\frac{\mathbf{a}'}{\|\mathbf{a}'\|}$.
\end{lemma}
Motivated by this discussion, we propose two transmit beamforming structures: (i) SSJB, where $\mathbf{X}$ combines sensing optimal without security and secure communication optimal precoders; and (ii) SLB, where $\mathbf{X}$ combines secure sensing optimal and secure communication optimal precoders, extendable to include the communication optimal design.
\subsection{SSJB Transmission Strategy}\label{firststrategy}
Based on the optimal solutions of Subproblems 1 and 2, in the SSJB scheme, the BS transmits:
\begin{align}
\mathbf{X}=\sqrt{P\tau}\mathbf{t}_1 \mathbf{s}_u+ \sqrt{P(1-\tau)}\mathbf{G} \mathbf{V} \in \mathbb{C}^{N\times L},\label{x}
\end{align}
where \(L > N\) represents the length of the radar pulse or communication frame; \(\mathbf{s}_u \in \mathbb{C}^{1 \times L}\) is a white Gaussian signaling data stream for the user with unit power, satisfying \(\frac{1}{L}E\{\mathbf{s}_u\mathbf{s}_u^H\} \approx 1\) when \(L\) is sufficiently large \cite{CrameRaoBoundOptimizationforJoint}. The matrix \(\mathbf{V} \in \mathbb{C}^{(N-2) \times L}\) represents an artificial noise (AN), whose elements are i.i.d. complex Gaussian, such that each row, denoted by \(\mathbf{v}_i\), satisfies \(\frac{1}{L}\mathbf{v}_i\mathbf{v}_i^H={\frac{1}{N-2}}\). We assume that the AN signals and data \(\mathbf{s}_u\) are orthogonal, i.e., \(\mathbf{s}_u\mathbf{v}_i^H=0\). The vectors \(\mathbf{t}_1\) and \(\mathbf{G}\) are constructed as follows. First, we construct an orthonormal basis for the \(N\)-dimensional space. The first basis vector is aligned with the target channel: \(\tilde{\mathbf{a}}=\frac{\mathbf{a}}{||\mathbf{a}||}\). The second basis vector is \(\tilde{\mathbf{h}}=\frac{\mathbf{h}-(\tilde{\mathbf{a}}^H \mathbf{h})\tilde{\mathbf{a}}}{||\mathbf{h}-(\tilde{\mathbf{a}}^H \mathbf{h})\tilde{\mathbf{a}}||}\). We choose \(\mathbf{t}_1\) to be a vector in the span of \(\tilde{\mathbf{a}}\) and \(\tilde{\mathbf{h}}\), i.e., \(\mathbf{t}_1= \alpha \tilde{\mathbf{a}} + \beta \tilde{\mathbf{h}}\). Next, defining the matrix \(\mathbf{A}= \begin{bmatrix} \mathbf{\tilde{a}} & \mathbf{\tilde{h}} \end{bmatrix} \in \mathbb{C}^{N \times 2}\), an orthonormal basis for the null space of \( \mathbf{A}^H \), i.e., all vectors \( \mathbf{x} \) such that \( \mathbf{A}^H \mathbf{x} = 0 \), is derived as follows. We compute the singular value decomposition (SVD) \(\mathbf{A}= \mathbf{U} \mathbf{\Sigma} \mathbf{\tilde{W}}^H\). Then, the last \(N - 2\) columns of \( \mathbf{U} \), denoted by \(\mathbf{t}_3,...,\mathbf{t}_N\), form the basis for the null space of \(\mathbf{A}^H\). The matrix \(\begin{bmatrix} \mathbf{\tilde{a}} & \mathbf{\tilde{h}} & \mathbf{G} \end{bmatrix}\), where \(\mathbf{G}=\begin{bmatrix} \mathbf{t}_3 & ... & \mathbf{t}_N \end{bmatrix}\), constructs an orthonormal basis for the \(N-2\)-dimensional space. The sample covariance matrix, due to the orthogonality of \(\mathbf{v}_i\) and \(\mathbf{s}_u\), is:  
\begin{align}
\mathbf{R}_x=\frac{1}{L}\mathbf {X}\mathbf{X}^H\approx 
P\tau\mathbf{t}_1\mathbf{t}^H_1+\frac{(1-\tau)P}{N-2}\sum_{i=3}^{N}\mathbf{t}_i\mathbf{t}^H_i.\label{Rx}
\end{align}  
Thus: $p_t= \text{Tr}(R_x) = P\tau\text{Tr}(\mathbf{t}_1\mathbf{t}_1^H) + \frac{(1-\tau)P}{N-2} \text{Tr} \left( \sum_{i=3}^{N} \mathbf{t}_i\mathbf{t}_i^H \right) \overset{(a)}{=} P.$ Here, (a) follows from the unitary property of \(\mathbf{U}\) and the orthonormality of \(\mathbf{\tilde{a}}\) and \(\mathbf{\tilde{h}}\), ensuring that \( P \) represents the BS transmit power, and \( \tau \) denotes the fraction allocated to $\mathbf{s}_u$.  
\subsection{SLB Transmission Strategy}\label{secondstrategy}
Based on the optimal solutions of Subproblems 2 and 3, in
the SLB scheme, the BS transmits:
\begin{align}
&\mathbf{X}=\sqrt{P\tau_1}\frac{\mathbf{h}}{||\mathbf{h}||} \mathbf{s}_u+ \sqrt{P\tau_2}\tilde{\mathbf{G}} \tilde{\mathbf{V}}+\sqrt{P\tau_3}\frac{\mathbf{a}}{||\mathbf{a}||} \mathbf{s}_{r1}\nonumber\\
&+\sqrt{P(1-\tau_3-\tau_2-\tau_1)}\frac{\mathbf{a}'}{||\mathbf{a}'||} \mathbf{s}_{r2} \in \mathbb{C}^{N\times L},\label{x2}
\end{align}
where the matrix \(\begin{bmatrix} \frac{\mathbf{h}}{||\mathbf{h}||} & \tilde{\mathbf{G}} \end{bmatrix}\) forms an orthonormal basis for the \(N\)-dimensional space. The radar signals \(\mathbf{s}_{r1}, \mathbf{s}_{r2} \in \mathbb{C}^{1 \times L}\) have unit power, satisfying \(\frac{1}{L}\mathbb{E}\{\mathbf{s}_{ri}\mathbf{s}_{i}^H\} = 1\) for $i \in \{r1,r2\}$. The artificial noise (AN) matrix \(\tilde{\mathbf{V}} \in \mathbb{C}^{(N-1) \times L}\) consists of i.i.d. complex Gaussian entries, where each row \(\tilde{\mathbf{v}}_i\) meets the power constraint \(\frac{1}{L}\tilde{\mathbf{v}}_i\tilde{\mathbf{v}}_i^H = \frac{1}{N-1}\). We assume orthogonality between the AN (\(\tilde{\mathbf{V}}\)), the data signal (\(\mathbf{s}_u\)), and the radar signals (\(\mathbf{s}_{r1}, \mathbf{s}_{r2}\)), such that: $\mathbf{s}_i\tilde{\mathbf{v}}_k^H = 0 \quad k \in \{1,...,N-1\} \quad \text{and} \quad \mathbf{s}_i\mathbf{s}_j^H = 0 \quad \text{for} \quad i \neq j \in \{r1, r2, u\}.$ Moreover, we have $p_t= \text{Tr}(R_x) = P$.
\section{CCDF of CRB at the BS}
In this section, we derive $P(\text{CRB}(\theta) > \epsilon)$ for the SSJB strategy in Lemma \ref{exactderivationlemma} (with the proof provided in Appendix \ref{proofexactderivationlemma}) and for the SLB strategy, in Lemma \ref{exactderivationlemma2} (proof in Appendix \ref{proofexactderivationlemma2}).
\begin{lemma}\label{exactderivationlemma}
In the SSJB strategy $\text{CRB}(\theta)$ is given as: $\text{CRB}(\theta)=Q[\gamma_1||\mathbf{b}'||^2N|\alpha|^2+\gamma_2M(||\mathbf{a}'||^2-\frac{(\sum_{i=1}^{N}-f'_i{t}_i)^2+(\sum_{i=1}^{N}f'_i{r}_i)^2}{K-\frac{1}{N}(T^2+R^2)}))]^{-1}$ where $\mathbf{b}'$ is the derivation of $\mathbf{b}(\theta)$ with respect to $\theta$, and \( f'_i=\pi \cos(\theta)\frac{N-(2i-1)}{2} \), ${r}_i\triangleq \mathcal{R}(e^{jf_i}{h}_i)=|{h}_i|\cos(f_i+{{\phi}}_i)$, ${t}_i\triangleq \mathcal{I}(e^{jf_i}{h}_i)=|{h}_i|\sin(f_i+{{\phi}}_i)$, ${{k}}_i\triangleq |{h}_i|^2$, ${R}\triangleq \sum_{i=1}^{N}{r}_i$, ${T}\triangleq \sum_{i=1}^{N}{t}_i$, ${{K}}\triangleq \sum_{i=1}^{N}{{k}}_i$, $||\mathbf{b}'||^2=\frac{\pi^2\cos^2(\theta)M(M^2-1)}{12}$, $||\mathbf{a}'||^2=\frac{\pi^2\cos^2(\theta)N(N^2-1)}{12}$, $Q\triangleq\frac{\sigma^2_R}{2 \mid c_3 \mid ^2 L}$, $P\tau \triangleq \gamma_1$ and $\frac{(1-\tau)P}{N-2} \triangleq \gamma_2$. Moreover, a lower bound for $P(\text{CRB}(\theta)>\epsilon)$ is given by 
 $P_{Lc}(\epsilon)=\!\!\frac{2}{\pi}\sin^{-1}(\sqrt{6}\sigma_r(\epsilon MN \pi ^2LP \mid c_3 \mid ^2\big(|\alpha|^2\tau(M^2-1)+\frac{(N^2-1)(1-\tau)}{N-2})\big)^{-1/2})$,
when $\sqrt{6}\sigma(\epsilon MN \pi ^2LP \mid c_3 \mid ^2\big(|\alpha|^2\tau(M^2-1)+\frac{(N^2-1)(1-\tau)}{N-2})\big)^{-1/2})<1$, and $0$ otherwise. Moreover, an approximation for $P(\text{CRB}(\theta)>\epsilon)$, $P_{Ac}$, is $P_{Ac}(\epsilon)=\frac{1}{\pi}\!\!\int_{0}^{\pi}\!\!\iiint_{\mathcal{\tilde{D}}(\theta, R,T, K)}\!\!\!f(R,T,K)\,dR\,dT\,dK,d\theta,$ where $\mathcal{\tilde{D}}(\theta, R, T, K)= \frac{\frac{6\sigma_R^2}{L|c_3|^2\pi^2\cos^2(\theta)MN}}{\gamma_1|\alpha|^2(M^2-1)+\gamma_2(N^2-1)(1-\frac{1}{K-\frac{1}{N}(T^2+R^2)}))}>\epsilon$, and $f(R,T, K)$ is the PDF of a trivariate normal distribution with a mean vector of $N\mathbf{\mu_d}$ and a covariance matrix of $N\mathbf{\Sigma_d}$ derived at Lemma \ref{lemma1i}.
\end{lemma}
\begin{lemma}\label{exactderivationlemma2}
In the SLB strategy $\text{CRB}(\theta)$ is equal to (\ref{crbsimplified}).
\begin{figure*}[t]
\normalsize
\begin{align}
\frac{QK}{||\mathbf{b}'||^2(y(T^2+R^2)+xK)+||\mathbf{b}||^2(\frac{|\mathbf{a}'|^2P\tau_2K}{N-1}+PK(1-\tau_1-\tau_2-\tau_3)|\mathbf{a}'|^2+\frac{xyK((\sum_{i=1}^{N}-f'_i{t}_i)^2+(\sum_{i=1}^{N}f'_i{r}_i)^2))}{(y(R^2+T^2))+xK)})}\label{crbsimplified}
\end{align}
\hrulefill
\end{figure*}
where $y\triangleq (P\tau_1-\frac{P\tau_2}{N-1})$, $x\triangleq NP\tau_3+\frac{NP\tau_2}{N-1}$. Moreover, an upper bound for $P(\text{CRB}(\theta)>\epsilon)$ in  SLB is given by:  
\begin{align}  
P_{Uc}(\epsilon)&=\frac{1}{\pi}\!\!\int_{-\pi/2}^{\pi/2}\!\!\iiint_{\mathcal{D}(\theta, R,T, K)}\!\!\!\!\!\!f(R,T,K)\,dR\,dT\,dK\,d\theta, \label{crbb1one}  
\end{align}  
where $\mathcal{D}(\theta, R, T, K)=QK[||\mathbf{b}'||^2(y(T^2+R^2)+xK)+M(\frac{|\mathbf{a}'|^2P\tau_2K}{N-1}+PK(1-\tau_1-\tau_2-\tau_3)|\mathbf{a}'|^2)]^{-1}>\epsilon$. A lower bound for $P(\text{CRB}(\theta)>\epsilon)$ is given by replacing $\mathcal{{D}}$ with $\mathcal{\tilde{D}}(\theta, R, T, K)
= QK[||\mathbf{b}'||^2(y(T^2+R^2)+xK)+M|\mathbf{a}'|^2K(\frac{P\tau_2}{N-1}+P(1-\tau_1-\tau_2-\tau_3)+\frac{xyK}{(y(R^2+T^2))+xK)})]^{-1}>\epsilon$ at (\ref{crbb1one}).  An approximation for $P(\text{CRB}(\theta)>\epsilon)$ is derived by replacing $\mathcal{\tilde{\tilde{D}}}(\theta, R, T, K) =QK[||\mathbf{b}'||^2(y(T^2+R^2)+xK)+||\mathbf{b}||^2K|\mathbf{a}'|^2(\frac{P\tau_2}{N-1}+P(1-\tau_1-\tau_2-\tau_3)+\frac{xy}{(y(R^2+T^2))+xK)})]^{-1}>\epsilon$ instead of $\mathcal{{D}}(\theta, R, T, K)$ into (\ref{crbb1one}).
\end{lemma}
In general, integrating a multivariate normal PDF over an arbitrary interval lacks a closed-form solution. However, we have used ray-tracing method and the MATLAB toolbox presented in \cite{Amethodtointegrate}, which enables the integration of Gaussian distributions in any dimension, over any domain.
\section{CCDF of CRB at the strong sensing eavesdropper}
In this section, we assume a strong eav—one that knows $\mathbf{s}_u$ and $\mathbf{V}$ in the SSJB strategy, and $\mathbf{s}_u$, $\tilde{\mathbf{V}}$, $\mathbf{s}_{r1}$, and $\mathbf{s}_{r2}$ in the SLB strategy. We derive $P(\text{CRB}(\phi) > \epsilon)$ for the SSJB strategy in Lemma \ref{crbsensingeav} and for the SLB strategy in Lemma \ref{crbsensingeav2}.
\begin{lemma}\label{crbsensingeav}
In the SSJB strategy $\text{CRB}(\phi) = \frac{\sigma^2_R}{2 |c_4|^2 L \gamma_1 ||\mathbf{c}'||^2 |\alpha|^2 N}$ where $||\mathbf{c}'||^2 = \frac{\pi^2 \cos^2(\phi) N_e (N_e^2 - 1)}{12}$ and $P(\text{CRB}(\phi)>{\epsilon})=\frac{2}{\pi}\sin^{-1}(\sqrt{6}\sigma_R({\epsilon} N_eN \pi ^2LP \mid c_4 \mid ^2|\alpha|^2\tau(N_e^2-1))^{-1/2})$ when $(\sqrt{6}\sigma_R({\epsilon} N_eN \pi ^2LP \mid c_4 \mid ^2|\alpha|^2\tau(N_e^2-1))^{-1/2})<1$, and $1$ otherwise.
\end{lemma} 
\begin{proof}
 Substituting (\ref{Rx}) into $\text{CRB}(\phi)$, as derived in Subsection \ref{opoftargeteav}, and using the orthogonality between $\mathbf{a}$ and $\mathbf{t}_i$ for $i = 3, \dots, N$ as well as $\tilde{\mathbf{h}}$, along with the condition $\mathbf{c'}^H(\phi)\mathbf{c}(\phi) = 0$ and the distribution of $\phi$, the proof is complete.
\end{proof}
\begin{lemma}\label{crbsensingeav2}
In the SLB strategy $P(\text{CRB}(\phi)>\epsilon)$ is given by replacing $\theta$ with $\phi$ and $\mathcal{{D}}$ with $\mathcal{\hat{D}}(\phi, R, T, K)
= \frac{\sigma_R^2 K}{2 |c_4|^2 L \, \|\mathbf{c}'\|^2 \, (y(R^2 + T^2) + xK)}>\epsilon$ at (\ref{crbb1one}).
\end{lemma} 
\begin{proof}
By substituting $\mathbf{a}^H{\mathbf{R}}_x\mathbf{a}$, as derived in Appendix \ref{proofexactderivationlemma2}, into the expression for $\text{CRB}(\phi)$ from Subsection \ref{opoftargeteav}, and rewriting$\mathbf{a}$ and $\mathbf{h}$ element-wise, we obtain: $
\text{CRB}(\phi) = \frac{\sigma_R^2 K}{2 |c_4|^2 L \, \|\mathbf{c}'\|^2 \, (y(R^2 + T^2) + xK)},
$ where $R$, $T$, and $K$ are defined in Lemma \ref{exactderivationlemma}. Then, following the same approach as in Appendix \ref{proofexactderivationlemma}, and noting that $\text{CRB}(\phi)$ is independent of $\theta$—since, by Lemma \ref{lemma1i}, the distributions of $R$, $T$, and $K$ are independent of $\theta$, and all other parameters in $\text{CRB}(\phi)$ are also independent of $\theta$—and given that $\phi$ is distributed as $\mathcal{U}(-\pi/2, \pi/2)$, the proof is complete.
\end{proof}
\section{CCDF of CRB at the weak sensing eavesdropper}
In this section, we assume a weak eav and derive $P(\text{CRB}(\phi) > \epsilon)$ for the SSJB strategy in Lemma \ref{crbsensingeavweak} (with the proof provided in Appendix \ref{proofcrbsensingeavweak}) and for the SLB strategy in \ref{crbsensingeav2weak} (with the proof provided in Appendices \ref{proofcrbsensingeav2weak}). 
\begin{lemma}\label{crbsensingeavweak}
In the SSJB strategy $\text{CRB}(\phi) =  \frac{6\sigma_r^2(\sigma_r^2 + |c_4|^2 L d N_e)}{|c_4|^4 L^3 d^2 \pi^2 \cos^2(\phi) N^2_e (N^2_e - 1)}$ where $d=P\tau |\alpha|^2N$ and $P(\text{CRB}(\phi)>{\epsilon})=\frac{2}{\pi}\sin^{-1}(\sqrt{6}\sigma_R(\epsilon|c_4|^4 L^3 d^2 \pi^2 N^2_e (N^2_e - 1))^{-1/2}(\sigma_r^2 + |c_4|^2 L d N_e)^{1/2})$ when $\sqrt{6}\sigma_R(\epsilon|c_4|^4 L^3 d^2 \pi^2  N^2_e (N^2_e - 1))^{-1/2}(\sigma_r^2 + |c_4|^2 L d N_e)^{1/2}<1$, and $1$ otherwise.
\end{lemma} 
\begin{lemma}\label{crbsensingeav2weak}
In the SLB strategy $P(\text{CRB}(\phi)>\epsilon)$ is given by replacing $\theta$ with $\phi$ and $\mathcal{{D}}$ with $\mathcal{\hat{D}}(\phi, R, T, K)
= \frac{6\sigma_r^2(\sigma_r^2 + |c_4|^2 L \tilde{d} N_e)}{|c_4|^4 L^3 \tilde{d}^2 \pi^2 \cos^2(\phi) N^2_e (N^2_e - 1)}>\epsilon$, where $\tilde{d}=y \frac{R^2+T^2}{K} + x$, at (\ref{crbb1one}).
\end{lemma} 
\section{Target Ergodic CRB at the BS and Sensing Eavesdropper}\label{targetergodiccrbatthebs}
In this section, we derive $\mathbb{E}[\text{CRB}(\theta)]$ and $\mathbb{E}[\text{CRB}(\phi)]$ (for the weak and strong eav) for the SSJB strategy in lemma \ref{newSLB} (proof in Appendix \ref{proofnewSLB})) and for the SLB strategy, in lemma \ref{forlargen}, where $\mathcal{I}(\mathcal{M}(S,K))\triangleq \frac{1}{N \sqrt{2 \pi N}} \int_{0}^{N + 5\sqrt{N}} \int_{0}^{10N} \mathcal{M}(S,K) e^{-\frac{s}{N} - \frac{(k - N)^2}{2N}} \, dK \, dS$.
\begin{lemma}\label{newSLB}
In the SSJB strategy, a lower bound of $\mathbb{E}[\text{CRB}(\theta)]$, is $\frac{12\sigma^2_R}{(\pi-2\delta)MN \pi ^2LP \mid c_3 \mid ^2\big(|\alpha|^2\tau(M^2-1)+\frac{(N^2-1)(1-\tau)}{N-2}\big)} \tan\left(\frac{\pi}{2} - \delta\right)$ and an approximation is given by
$\frac{2 \cot(\delta)}{(\pi - 2 \delta)} \mathcal{I}(\mathcal{M}_1(S,K))$ where $\mathcal{M}_1(S,K)=\frac{6\sigma_R^2[\gamma_1|\alpha|^2(M^2-1)+\gamma_2(N^2-1)(1-\frac{1}{K-\frac{S}{N}}))]^{-1}}{L|c_3|^2\pi^2MN}$\footnote{In Section \ref{simulations}, we illustrate through numerical results that this approximation closely aligns with the exact value of $\mathbb{E}[\text{CRB}(\theta)]$ and also the lower bound is tight.}. Moreover, the closed-form expression for the ergodic CRB at the strong and weak sensing eav are:
$\mathbb{E}[\text{CRB}(\phi)] = 
\frac{12\sigma_R^2}{N_e N \pi^2 (\pi-2\delta) L P |c_4|^2 |\alpha|^2 \tau (N_e^2 - 1)} \tan\left(\frac{\pi}{2} - \delta\right)
$ and $\mathbb{E}[\text{CRB}(\phi)] = 
\frac{12\sigma_r^2(\sigma_r^2 + |c_4|^2 L d N_e)}{|c_4|^4 L^3 d^2 \pi^2  N^2_e (N^2_e - 1)(\pi-2 \delta)} \tan\left(\frac{\pi}{2} - \delta\right)
$, where $d=P\tau |\alpha|^2N$, respectively.
\end{lemma}
\begin{lemma}\label{forlargen}
In the SLB strategy, an upper bound of $\mathbb{E}[\text{CRB}(\theta)]$, is $ \frac{2 \cot(\delta)}{(\pi - 2 \delta)} \mathcal{I}(\mathcal{M}_2(S,K)).
$ where $\mathcal{M}_2(S,K)=QK[\frac{\pi^2 M(M^2 - 1)}{12} (yS + xK) + M \left( \frac{\frac{\pi^2 N(N^2 - 1)}{12} P \tau_2 K}{N - 1} \right)]^{-1}$.
A lower bound of $\mathbb{E}[\text{CRB}(\theta)]$ and an approximation, is given by replacing $\mathcal{{M}}_2(S,K)$ with $QK[\frac{\pi^2 M(M^2 - 1)}{12}(yS+xK)+M\frac{\pi^2 N(N^2 - 1)}{12}K(\frac{P\tau_2}{N-1}+\frac{xyK}{(yS+xK)}]^{-1}$ and $QK[\frac{\pi^2 M(M^2 - 1)}{12}(yS+xK)+MK\frac{\pi^2 N(N^2 - 1)}{12}(\frac{P\tau_2}{N-1}+\frac{xy}{(yS)+xK)})]^{-1}$, respectively \footnote{In Section \ref{simulations}, we illustrate through numerical results that the approximation closely aligns with the exact value of $E[\text{CRB}(\theta)]$ and also the upper bounds are tight.}. Moreover, for the strong and weak eav, $ \mathbb{E}[\text{CRB}(\phi)]$ is derived by replacing $\mathcal{{M}}_2(S,K)$ with $\sigma_R^2K[2 |c_4|^2 L \, \frac{\pi^2 N_e(N_e^2 - 1)}{12} \, (yS + xK)]^{-1}$ and $[6\sigma_r^2(\sigma_r^2 + |c_4|^2 L (y \frac{S}{K} + x) N_e)] \times [|c_4|^4 L^3 (y \frac{S}{K} + x)^2 \pi^2 N^2_e (N^2_e - 1)]^{-1}$, respectively.
\end{lemma}
\begin{proof}
Using the derived lower bound, upper bound, and approximation of $\text{CRB}(\theta)$ from Appendix \ref{proofexactderivationlemma2}, along with the expressions for $\text{CRB}(\phi)$ for the strong and weak sensing eav, as given in Lemmas \ref{crbsensingeav2} and \ref{crbsensingeav2weak}, and by following the same approach used in Appendix \ref{proofnewSLB}, the ergodic CRBs $\mathbb{E}[\text{CRB}(\theta)]$ and $\mathbb{E}[\text{CRB}(\phi)]$ are derived.
\end{proof}
\section{Ergodic Secrecy Rate}\label{secrecyergodicrate}
In this section, we derive the ESR. For an external communication eav, it is given by $
C_s = \left( \mathbb{E}_{\mathbf{h},\, \mathbf{a}(\theta)} \left[ \log\left(1 + \text{SINR}_u\right) \right] - \mathbb{E}_{\mathbf{h},\, \mathbf{a}(\theta),\, \mathbf{h}_e} \left[ \log\left(1 + \text{SINR}_e\right) \right] \right)^+,
$ and when the target itself is the eav, by $
C_s = \left( \mathbb{E}_{\mathbf{h},\, \mathbf{a}(\theta)} \left[ \log(1 + \text{SINR}_u) \right] - \mathbb{E}_{\mathbf{h},\, \mathbf{a}(\theta)} \left[ \log(1 + \text{SINR}_t) \right] \right)^+,
$
as discussed in Section \ref{metric}. We first derive the user’s ergodic rate, the information leakage to an external eav, and the leakage to a malicious target under the SSJB strategy. These are presented in Lemmas \ref{lemmapdfuser} (proof in Appendix \ref{proofpdfuser}), \ref{lemmapdfeav} (proof in Appendix and \ref{proofpdfeav}), and \ref{malitarget}, respectively. Next, we provide corresponding results for the SLB strategy in Lemmas \ref{lemmapdfuser2}, \ref{lemmapdfeav2}, and \ref{malitarget2}.
\begin{lemma}\label{lemmapdfuser}
In the SSJB strategy, two upper bound for $E_{\mathbf{h,a}} \left[ \log(1 + \text{SINR}_u) \right]$ are: $\log\left(1 + \frac{P \tau |c_1|^2}{\sigma_u^2} (|\alpha|^2 + |\beta|^2 (N - 1))\right)$ \footnote{In Section \ref{simulations}, we illustrate through numerical results that this bound is tight and closely aligns with the exact value of $\mathbb{E}_{\mathbf{h},\, \mathbf{a}(\theta)} \left[ \log\left(1 + \text{SINR}_u\right) \right]$.} and $\frac{e^{\sigma_u^2/P \tau |c_1|^2}}{(N-1)!} \cdot G_{2,3}^{3,1}\left(a \middle| \begin{array}{c} 
0, 0 \\ 
0, -1, N-1 
\end{array} \right),$
where \( G \) is the Meijer G-function.
\end{lemma}
\begin{lemma}\label{lemmapdfeav}  
In the SSJB strategy, the exact ergodic rate at the communication eav is given by:  
\begin{align}
E_{\mathbf{h_e,h},\mathbf{a}(\theta)}\left[ \log(1 + \text{SINR}_e) \right]\!\!&\overset{(a)}{=}\!\!\!\int_{0}^{\infty}\!\!\!\!\! e^{(-\frac{T}{2C_1})}(1 + \frac{TC_2}{C_1})^{-(N-2)}dt,\label{final2}
\end{align}  
where \(T=2^t-1\), $C_1 = \frac{P \tau |c_2|^2}{\sigma^2}$, and $C_2 = \frac{|c_2|^2P(1-\tau)/(N-2)}{\sigma^2}.$ 
\end{lemma}  
We note that (\ref{final2}) can be expressed in terms of the incomplete Gamma functions; however, these representations are not computationally simpler than (\ref{final2}) itself, as it can be efficiently evaluated numerically using a single integral function in MATLAB.
\begin{lemma}\label{malitarget}
In the SSJB strategy, the exact ergodic leakage rate to the malicious target is $\log(1+\frac{P\tau|c_5|^2|\alpha|^2 N}{\sigma^2_t})$.
\end{lemma}
\begin{proof}
Based on (\ref{x}), and given the signal received at the taget as described in Section \ref{systemmodel}, along with the orthogonality of \(\mathbf{v}_i\) and \(\mathbf{s}_u\), and the orthogonality of \(\mathbf{a}\) and \(\mathbf{t}_i\), the SINR at the target is given by: $\text{SINR}_t=\frac{P\tau|c_5|^2}{\sigma^2_t}|\mathbf{a}^H\mathbf{t}_1|^2=\frac{P\tau|c_5|^2|\alpha|^2 N}{\sigma^2_t}$ where the second equality is due to the orthogonality between $\mathbf{a}$ and $\tilde{\mathbf{h}}$.
\end{proof}
\begin{lemma}\label{lemmapdfuser2}
In the SLB strategy, upper bound of $\mathbb{E}_{\mathbf{h},\, \mathbf{a}(\theta)} \left[ \log\left(1 + \text{SINR}_u\right) \right]$ is $\mathcal{I}(\mathcal{M}_3(S,K))$ where $\mathcal{M}_3(S,K)= \log_2(1+\frac{P |c_1|^2 \tau_1 K}{\sigma^2 + |c_1|^2P \tau_3 \frac{S}{N}}),$. Lower bound and approximation for $\mathbb{E}_{\mathbf{h},\, \mathbf{a}(\theta)} \left[ \log\left(1 + \text{SINR}_u\right) \right]$ are derived by substituting the expression $\mathcal{M}_3(S, K)$ with $\log_2(1+P |c_1|^2 \tau_1 K[\sigma^2 + |c_1|^2P \tau_3 \frac{S}{N} + |c_1|^2P(1 - \tau_3 - \tau_2 - \tau_1) K]^{-1}),$ and $\log_2(1+P |c_1|^2\tau_1 K[\sigma^2 + P |c_1|^2\tau_3 \frac{S}{N} +|c_1|^2 P(1 - \tau_3 - \tau_2 - \tau_1)]^{-1})$, respectively \footnote{In Section \ref{simulations}, we illustrate through numerical results that the approximation closely aligns with the exact value of $\mathbb{E}_{\mathbf{h},\, \mathbf{a}(\theta)} \left[ \log\left(1 + \text{SINR}_u\right) \right]$ and also the upper bounds are tight.}.
\end{lemma}
\begin{proof}
Based on equation \eqref{x2}, and considering the received signal at the user as described in Section \ref{systemmodel}, along with the orthogonality between $\mathbf{\tilde{G}}$ and $\mathbf{h}$, the signal-to-interference-plus-noise ratio (SINR) at the user is given by:
$
\text{SINR}_u = \frac{P |c_1|^2 \tau_1 |\mathbf{h}|^2}{\sigma^2 + |c_1|^2P \tau_3 \frac{|\mathbf{h}^H \mathbf{a}|^2}{\|\mathbf{a}\|^2} + |c_1|^2 P(1 - \tau_3 - \tau_2 - \tau_1) \frac{|\mathbf{h}^H \mathbf{a}'|^2}{\|\mathbf{a}'\|^2}}.
$
By expressing this in element-wise form and applying the definitions from Lemma \ref{exactderivationlemma2}, the SINR simplifies to:
$
\text{SINR}_u = \frac{P |c_1|^2 \tau_1 K}{\sigma^2 +|c_1|^2 P \tau_3 \frac{R^2 + T^2}{N} +|c_1|^2 P(1 - \tau_3 - \tau_2 - \tau_1) \cdot \frac{(\sum_{i=1}^{N} -f'_i t_i)^2 + (\sum_{i=1}^{N} f'_i r_i)^2}{\|\mathbf{a}'\|^2}},
$
where $K $, $T $, and $R$, are defined in Lemma \ref{exactderivationlemma2}. Then, following the approach used in Appendix \ref{proofexactderivationlemma2} and Appendix \ref{proofnewSLB}, the proof is complete and we ignore the detail for brevity.
\end{proof}
\begin{lemma}\label{lemmapdfeav2}  
In the SLB strategy, an approximation of the ergodic rate at the communication eav is given by $log2(1+\frac{|c_2|^2P \tau_1}{\sigma^2 + |c_2|^2(P-P\tau_1)}).$ \footnote{In Section \ref{simulations}, we illustrate through numerical results that the approximation closely aligns with the exact value of $\mathbb{E}_{\mathbf{h},\, \mathbf{a}(\theta),\, \mathbf{h}_e} \left[ \log\left(1 + \text{SINR}_e\right) \right]$.}.
\end{lemma}
\begin{proof}
Based on (\ref{x2}), and given the received signal at the communication eav as described in Section \ref{systemmodel}, its SINR is:
$\text{SINR}_e = P |c_2|^2 \tau_1 \frac{|\mathbf{h}^H_e\mathbf{h}|^2}{|\mathbf{h}|^2}[\sigma^2 + |c_2|^2 P \tau_3 \frac{|\mathbf{h}_e^H \mathbf{a}|^2}{\|\mathbf{a}\|^2} + |c_2|^2P(1 - \tau_3 - \tau_2 - \tau_1) \frac{|\mathbf{h}_e^H \mathbf{a}'|^2}{\|\mathbf{a}'\|^2}+\frac{P|c_2|^2\tau_2}{N-1}\mathbf{h}^H_e(\mathbf{I}-\frac{\mathbf{h}\mathbf{h}^H}{|\mathbf{h}|^2})\mathbf{h}_e]^{-1}.$ 
Since \(\mathbf{h}_e\) is independent of \(\mathbf{h}\), \(\mathbf{h}^H_e \mathbf{h} \sim \mathcal{CN}(0, \|\mathbf{h}\|^2)\). Thus, \(\frac{|\mathbf{h}^H_e \mathbf{h}|^2}{\|\mathbf{h}\|^2} \sim \text{Exp}(1)\) (exponential with mean 1) and $\mathbb{E}\left[ P \tau_1 \frac{|\mathbf{h}^H_e \mathbf{h}|^2}{\|\mathbf{h}\|^2} \right] = P \tau_1.$ Moreover, \(\mathbf{h}_e^H \mathbf{a} \sim \mathcal{CN}(0, N)\), so \(\frac{|\mathbf{h}_e^H \mathbf{a}|^2}{N} \sim \text{Exp}(1)\). \(\mathbf{h}_e^H \mathbf{a}' \sim \mathcal{CN}(0, \|\mathbf{a}'\|^2)\), so \(\frac{|\mathbf{h}_e^H \mathbf{a}'|^2}{\|\mathbf{a}'\|^2} \sim \text{Exp}(1)\). Moreover, the matrix \(\mathbf{I} - \frac{\mathbf{h} \mathbf{h}^H}{\|\mathbf{h}\|^2}\) is a projection matrix with rank \(N-1\). So, \(\mathbf{h}^H_e \left( \mathbf{I} - \frac{\mathbf{h} \mathbf{h}^H}{\|\mathbf{h}\|^2} \right) \mathbf{h}_e \sim \text{Gamma}(N-1, 1)\) with mean $N-1$. Thus, we approximate: $\mathbb{E}[\text{SINR}_e] \approx \frac{\mathbb{E}[\text{Numerator}]}{\mathbb{E}[\text{Denominator}]} = \frac{P |c_2|^2\tau_1}{\sigma^2 + |c_2|^2P \tau_3 + |c_2|^2P(1 - \tau_3 - \tau_2 - \tau_1) + |c_2|^2P \tau_2}=\frac{P |c_2|^2\tau_1}{\sigma^2 + |c_2|^2(P-P\tau_1)}.$
\end{proof}
\begin{lemma}\label{malitarget2}
In the SLB strategy, the exact ergodic rate at the malicious target is $\mathcal{I}(\mathcal{M}_4(S,K))$ where $\mathcal{M}_4(S,K)= \log_2(1+\frac{P |c_5|^2 \tau_1 \frac{S}{K}}{\sigma^2 + |c_5|^2 P \tau_3 N + \frac{P |c_5|^2 \tau_2}{N-1} \left( N - \frac{S}{K} \right)})$
\end{lemma}
\begin{proof}
Based on equation (\ref{x2}) and the received signal at the malicious target described in Section \ref{systemmodel}, along with the orthogonality between $\mathbf{a}$ and $\mathbf{a'}$, we have:
$\text{SINR}_t = \frac{P |c_5|^2 \tau_1 \frac{|\mathbf{a}^H\mathbf{h}|^2}{|\mathbf{h}|^2}}{\sigma^2 + |c_5|^2 P \tau_3 N + \frac{P |c_5|^2 \tau_2}{N-1} \left( N - \frac{|\mathbf{a}^H \mathbf{h}|^2}{|\mathbf{h}|^2} \right)}.
$ By expressing this in element-wise form and applying the definitions from Lemma \ref{exactderivationlemma2}, and noting that the distributions of $R$, $T$, and $K$ are independent of $\theta$ (as established in Lemma \ref{lemma1i}), and following the same approach as in Appendix \ref{proofnewSLB}, the proof is complete.
\end{proof}
\section{Numerical Results}\label{simulations}
Unless stated otherwise, we use the following parameters: \( N = 15 \), \( M = 17 \), \( N_e = 15 \), \( P = 10 \), \( \sigma_u = \sigma_r = 1 \), \( L = 30 \), \( c_1 = c_2 = c_5=\sqrt{0.001} \), \( \delta = 0.1 \), and \( c_3 = c_4 = 0.001 \).  Monte Carlo simulation results are averaged over 10,000 channel realizations and closely match the numerical derivations across all figures.

Fig. \ref{ratesfigure} \subref{c12}–\subref{c13} and Fig. \ref{ratesfigure} \subref{c} illustrate the contours of the ESR (where the ESR for an external eav is denoted as \( C_s \) and shown with filled color, while for a malicious target, it is \( C^t_s \) and shown with lines) for SLB\footnote{We note that the axis ranges follow \( 0 < \tau_1 < 1 \), \( 0 < \tau_2 < 1 \), \( 0 < \tau_3 < 1 \), and \( 0 < 1 - \tau_1 - \tau_2 - \tau_3 < 1 \).} and SSJB, versus the BS power coefficients (\(\tau_1\) for \(\mathbf{s}_u\) (user data), \(\tau_2\) for \(\mathbf{V}\) (AN), and \(\tau_3\) for \(\mathbf{s}_{r1}\) (radar signal directed toward the target channel, i.e., \(\mathbf{a}\)) in SLB; and \(\tau\) (for \(\mathbf{s}_u\)) and \(\alpha\) (the fraction of \(\mathbf{s}_u\) directed toward \(\mathbf{a}\)) in SSJB\footnote{Since the ESR depends on two parameters in SSJB, a contour plot effectively captures the three-dimensional relationship of ESR as a function of \(\tau\) and \(\alpha\). Moreover the range of the parameters are $0<\alpha<1$ and $0<\tau<1$}. Additionally, Fig. \ref{ratesfigure} \subref{firstrate} and Fig. \ref{ratesfigure} \subref{secondrate} show the ergodic rates of the user and the communication eav (malicious target or external eav) at SSJB and SLB, respectively\footnote{The x-axis labels \(\tau_1 / \tau_2 / \tau_3\) indicate that, for example, when \(\tau_1\) and \(\tau_2\) (or \(\tau_3\)) are fixed, the rate is plotted against \(\tau_3\) (or \(\tau_2\)).}.  
\begin{figure*}
    \centering
    \subfigure[ESR, SLB, $\tau_3$=0.07 \label{c12}]{\includegraphics[scale=.42]{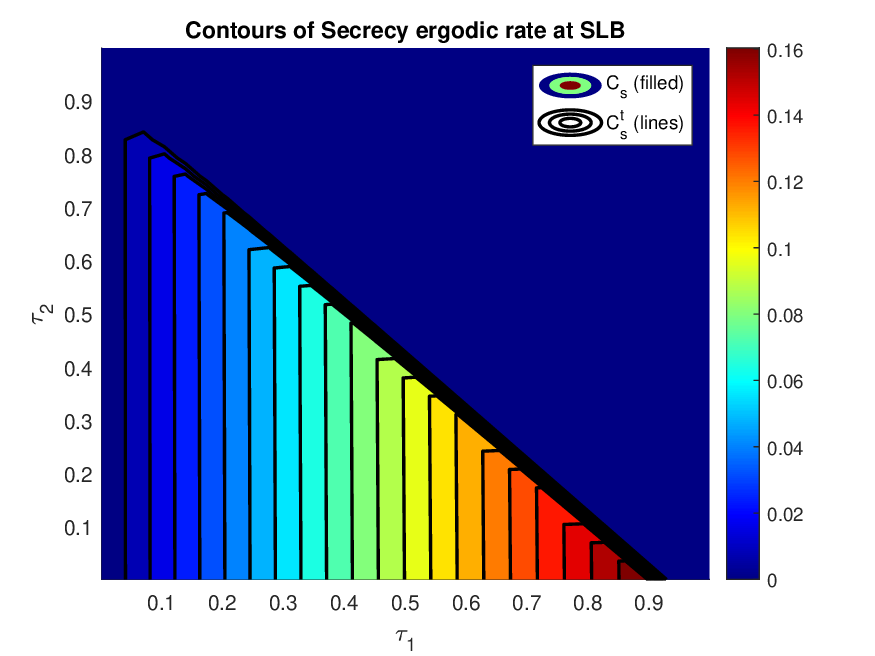}}\hspace{-0.5cm}
    \subfigure[ESR, SLB, $\tau_2$=0.07 \label{c13}]{\includegraphics[scale=.42]{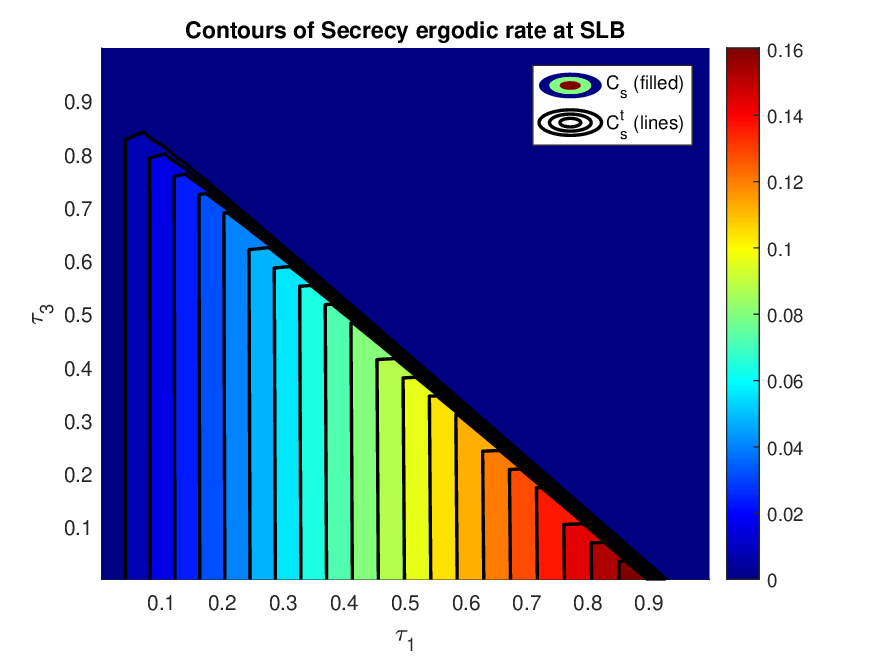}}\hspace{-0.5cm}
    \subfigure[ESR, SSJB \label{c}]{\includegraphics[scale=.42]{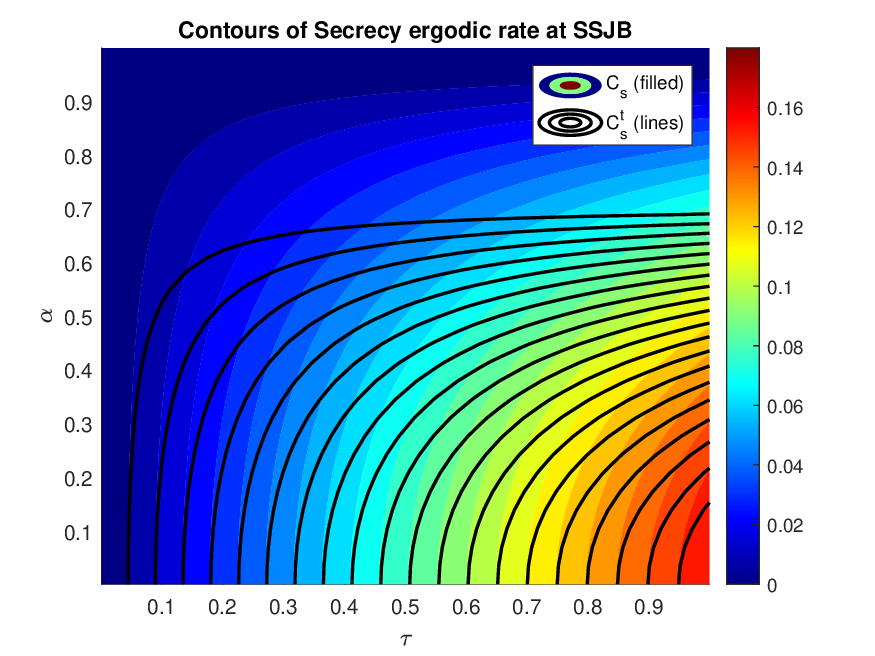}}\hspace{-0.5cm}
     \subfigure[SSJB, Ergodic rate of user and commu. eav \label{firstrate}]{\includegraphics[scale=.42]{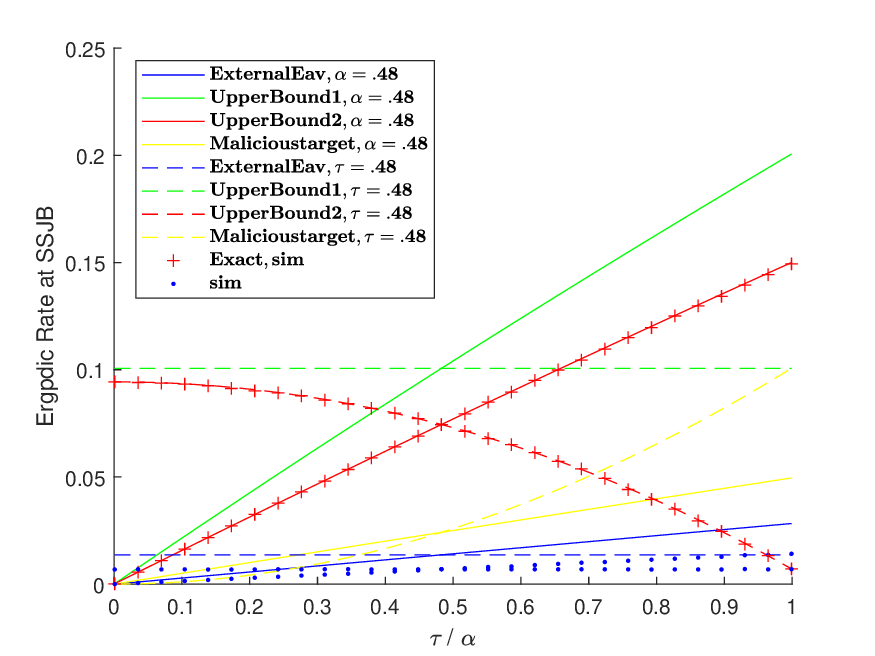}}\hspace{-0.5cm}
    \subfigure[SLB, Ergodic rate of user and commu. eav \label{secondrate}]{\includegraphics[scale=.42]{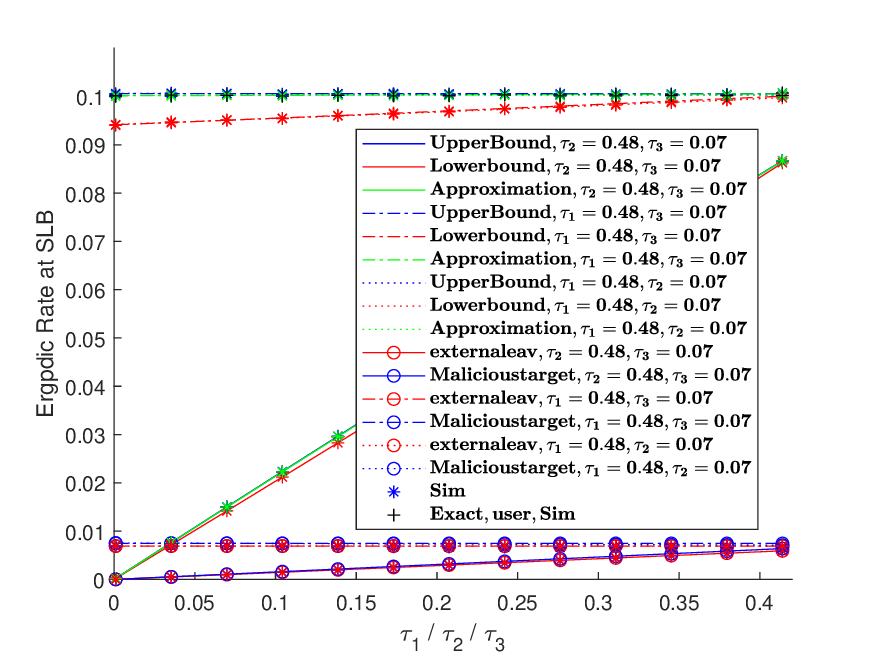}}\hspace{-0.5cm}
    \subfigure[SSJB and SLB, ESR at subproblems\label{subproblemrate}]
    {\includegraphics[scale=.42]{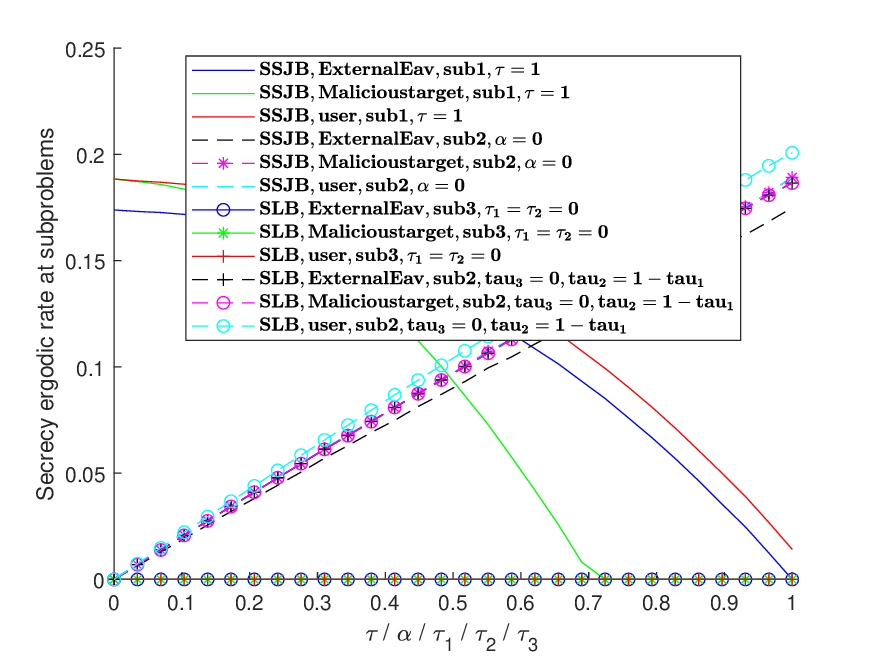}}\hspace{-0.5cm}
    \caption{Ergodic rates vs. BS power coefficients: SLB (\(\tau_1\): user data; \(\tau_2\): AN; \(\tau_3\): radar to target) and SSJB (\(\tau\): user data; \(\alpha\): user data fraction to target).}
    \label{ratesfigure}
\end{figure*}
From Fig. \ref{ratesfigure} \subref{c12}–\subref{c13}, it can be observed that in SLB, as \(\tau_1\) increases, the ESR also increases. The effects of \(\tau_2\) and \(\tau_3\) on the ESR are approximately the same—transmitting AN in the null space of the user channel has a similar impact as directing the radar signal toward the target, whether for a malicious target or an external eav. Furthermore, the ESR is more sensitive to changes in \(\tau_1\) than in \(\tau_2\) or \(\tau_3\) (the ESR remains nearly constant when \(\tau_2\) or \(\tau_3\) varies). This observation is also validated by Fig. \ref{ratesfigure} \subref{secondrate}, where the ergodic rates of both the user and the eav (target or external eav) increase with \(\tau_1\), though the user rate increases more significantly. Additionally, the rates for both the user and the eav remain approximately unchanged with increases in \(\tau_2\) or \(\tau_3\), with the user rate consistently higher for \(\tau_1 \gtrsim 0.08\), resulting in a positive ESR. Moreover, the rates for the malicious target and external eav are nearly identical, with the malicious target rate slightly higher. This further supports Fig. \ref{ratesfigure} \subref{c12}–\subref{c13}, where the contour regions for ESR are approximately the same. The approximation of the user rate closely matches its exact value, and the upper bound is almost tight. 

From Fig. \ref{ratesfigure} \subref{c}, it is evident that in SSJB, as \(\tau\) increases (or \(\alpha\) increases), the ESR (for both malicious target and external eav) increases (or decreases, respectively).
However, unlike SLB, the ESR contour regions for the malicious target and external eav differ. For example, when \(\alpha > 0.7\), the ESR for the malicious target drops to zero, whereas for the external eav, it reaches zero only when \(\alpha > 0.9\). This is also validated by Fig. \ref{ratesfigure} \subref{firstrate}, where the rates of both the malicious target and external eav increase with \(\alpha\) and \(\tau\), whereas the user rate increases with \(\tau\) (with a steeper increase than the eav’s) but decreases with \(\alpha\), leading to zero ESR. Additionally, the upper bound (Bound 1) is almost tight for the user rate, and the malicious target’s rate is higher than that of the external eav. Furthermore, the user rate’s upper bound (Bound 2) is also tight.

Furthermore, Fig. \ref{ratesfigure} \subref{subproblemrate} shows the ESR at subproblem points (denoted as "sub" in the legend). Specifically: $\bullet$ Subproblem 1 (sub1) in SSJB: At \(\tau = 1\) (no AN, only user data), we plot ESR versus \(\alpha\). When \(\alpha = 0\) (user data transmitted orthogonally to the target channel, \(\mathbf{\tilde{h}}\) in Sec. \ref{firststrategy}), the malicious target's ESR equals the user rate (no power received), while the external eav's ESR is lower (indicating leakage). As \(\alpha\) increases, the target's rate rises, and the malicious target's ESR drops, reaching zero before \(\alpha = 1\). The external eav's ESR reaches zero exactly at \(\alpha = 1\). 
$\bullet$ Subproblem 2 (sub2) in SSJB: At \(\alpha = 0\) (no user data toward the target, and AN in \(\mathbf{G}\) as in Sec. \ref{firststrategy}). When \(\tau = 0\) (AN only), all ESRs and user rates are zero, as expected.
$\bullet$ Subproblem 2 (sub2) in SLB: At \(\tau_3 = 0\) and \(\tau_2 = 1 - \tau_1\) (user data in \(\mathbf{{h}}\) like MRT beamformer and AN in \(\mathbf{\tilde{G}}\) as in Sec. \ref{secondstrategy}), we plot ESR versus \(\tau_1\). At \(\tau_1 = 0\) (AN only), all ESRs and user rates are zero. As \(\tau_1\) increases, both malicious target and external eav ESRs rise, reaching the same value at \(\tau_1 = 1\), while the user rate increases further.
$\bullet$ Subproblem 3 (sub3) in SLB: At \(\tau_1 = \tau_2 = 0\) (only radar signals toward \(\mathbf{a}\) and \(\mathbf{a}'\)), all ESRs and user rates remain zero since no user data is transmitted. Note that zero ESR at \(\tau = 1\) (sub1, SSJB) occurs because the BS beam aligns with the target. In all other cases, zero ESR results from zero power allocated to user data. Thus, we conclude that the SSJB strategy is more resilient against the malicious target; there is chance that when the user rate is positive, the malicious target rate becomes zero (indicating a positive ESR in this case). This is because SSJB can transmits user data orthogonally to the target channel ($\mathbf{\tilde{h}}$), which is not possible with SLB. However, SLB offers stronger resistance against external eavesdroppers and achieves higher rates in both the presence and absence of such eavesdroppers, as it transmits user data in the direction of the user channel (similar to MRT), whereas SSJB is constrained to the orthogonal component of the target channel, i.e., $\mathbf{\tilde{h}}$.

Fig. \ref{crbfigures} \subref{firstbs}, \subref{firstsensing}, and \subref{secondbssensing} show the ergodic CRB for the BS in SSJB, the strong/weak sensing eaves (eav) in SSJB, and the strong/weak sensing eav/ the BS in SLB, respectively. It can be observed that the ergodic CRB at the weak sensing eav is higher than that at the strong sensing eav, which in turn is higher than the ergodic CRB at the BS in both SLB and SSJB. Furthermore, Fig. \ref{crbfigures} \subref{firstbs} and \subref{secondbssensing} demonstrate that the approximation of the ergodic CRB at the BS closely matches its exact value in both SSJB and SLB, with the lower bound being tight in SSJB and the upper bound being tight in SLB. From Fig. \ref{crbfigures} \subref{firstsensing} and \subref{secondbssensing}, it is seen that increasing \(\alpha\) or \(\tau\) in SSJB, or \(\tau_1\), \(\tau_2\), or \(\tau_3\) in SLB, reduces the ergodic CRB at the sensing eav. However, Fig. \ref{crbfigures} \subref{firstbs} reveals that for \(\alpha < 0.2\), increasing \(\tau\) raises the ergodic CRB at the BS, whereas for \(\alpha > 0.2\), increasing \(\tau\) decreases it. Additionally, Fig. \ref{crbfigures} \subref{secondbssensing} indicates that increasing \(\tau_3\) (or \(\tau_1\), \(\tau_2\)) increases the CRB at the BS, though the CRB is more sensitive to changes in \(\tau_1\) or \(\tau_2\) than in \(\tau_3\).
\begin{figure*}
    \centering
    \subfigure[SSJB, Ergodic CRB($\theta)$ at BS \label{firstbs}]{\includegraphics[scale=.42]{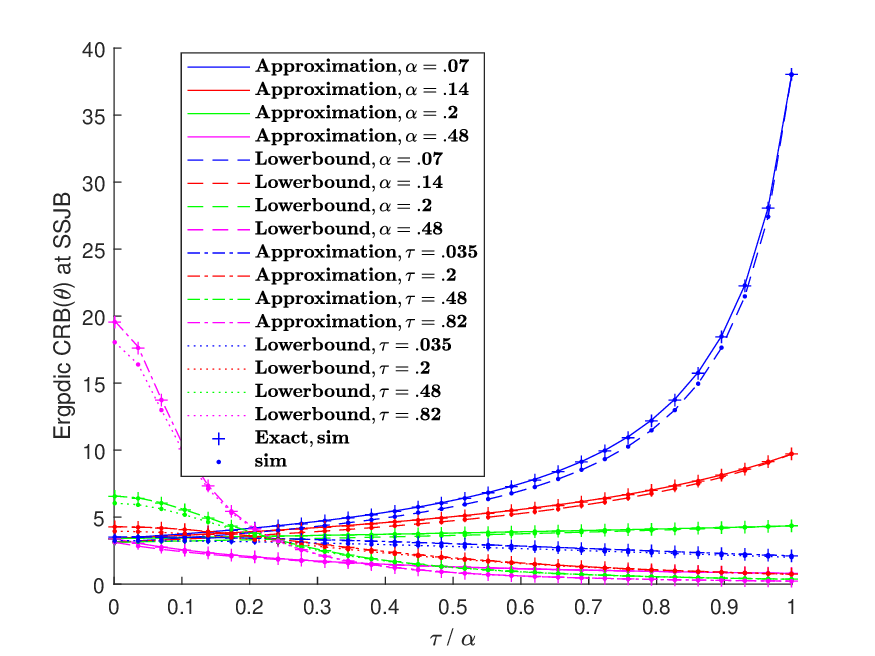}}\hspace{-0.5cm}
    \subfigure[SSJB, Ergodic CRB($\phi)$ at sensing eav \label{firstsensing}]{\includegraphics[scale=.42]{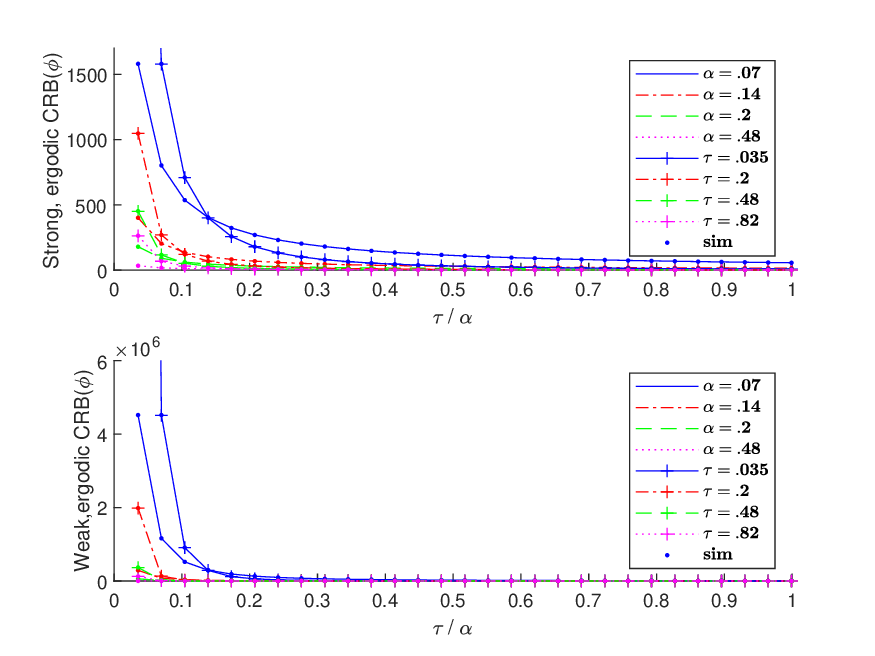}}\hspace{-0.5cm}
    \subfigure[SLB, Ergodic CRB at BS and sensing eav \label{secondbssensing}]{\includegraphics[scale=.42]{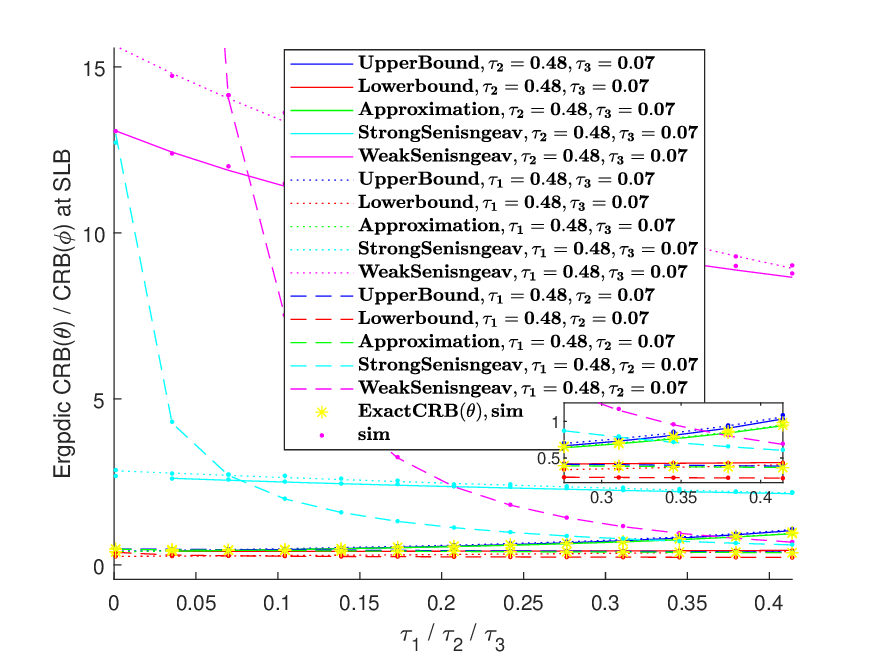}}\hspace{-0.5cm}
    \caption{Ergodic CRBs vs. BS power coefficients: SLB (\(\tau_1\): user data; \(\tau_2\): AN; \(\tau_3\): radar to target) and SSJB (\(\tau\): user data; \(\alpha\): user data fraction to target)."}
    \label{crbfigures}
\end{figure*}

Fig. \ref{regions} \subref{region1} and \subref{region2} show the trade off regions, ergodic CRB of the BS and the strong eav versus the ESR (in all three cases: external communication eav, malicious target, or no communication eav), respectively\footnote{We have zoomed in on the x-axis.}. In Fig. \ref{regions} \subref{region1}, the traditional tradeoff between S \& C in ISAC (when no communication eav is present) is illustrated. Additionally, when a communication eav exists, whether a malicious target or an external eav, a tradeoff still exists in a random secure communication ISAC system. Specifically, increasing the ESR (improving secure communication) leads to a higher ergodic CRB at the BS (reducing parameter estimation accuracy). Thus, achieving a positive ESR comes at the cost of sacrificing angle estimation accuracy. The sensing optimal point (SO), as discussed in prior works \cite{onthefundementaltradeoff,OnStochasticFundamentalLimitsinaDownlink}, which represents the minimum achievable sensing MSE regardless of communication performance, is also depicted. It can be observed that the minimum ergodic CRB is approximately the same for both SSJB and SLB. Furthermore, across all other points, including the communication optimal (CO) point,\footnote{The CO points for SSJB are not shown, as they are achieved at much higher values of $E[\text{CRB}(\theta)]$.} which maximizes rate regardless of sensing performance, SLB with no eav (i.e., traditional ISAC) consistently achieves a higher rate than SSJB in all scenarios. At high values of $E[\text{CRB}(\theta)] > 100$ (not shown in the plot), the ESR for SSJB in both the malicious target\footnote{Beyond $E[\text{CRB}(\theta)] \sim 1.5$, the malicious target’s rate surpasses that of the external eav in SSJB.} and no-eav cases increases and converges to nearly the same value. It even temporarily exceeds the ESRs of SLB in the malicious target and external eav cases, but remains below the SLB (no-eav) curve. Meanwhile, the ESR for the external eav case in SSJB remains consistently lower than all other curves. Notably, the SLB curve is entirely confined to $x < 2.5$. Moreover, the ESRs for the malicious target and external eav cases in SLB are nearly identical. These observations support the conclusions drawn from Fig.\ref{ratesfigure}\subref{subproblemrate}, namely that the SSJB strategy offers greater resilience against malicious targets, while SLB provides stronger protection in the presence of external eaves or when no eav is present. Comparing Fig. \ref{regions} \subref{region1} with Fig. \ref{regions} \subref{region2}, we observe that in both SLB and SSJB, the BS consistently achieves better angle estimation than the sensing eav, as \( E[\text{CRB}(\phi)] \) is higher than \( E[\text{CRB}(\theta)] \). Unlike Fig. \ref{regions} \subref{region1}, the regions in Fig. \ref{regions} \subref{region2} exhibit a decreasing trend (for all cases: no eav, external eav, or malicious target). This implies that as ESR increases (enhancing secure communication), the sensing eav’s ergodic CRB decreases (improving the eav’s estimation ability and thus reducing sensing privacy). This highlights a tradeoff between security and privacy in random ISAC networks. Based on these curves, system designers can determine how much sensitive data leakage is tolerable to ensure sensing privacy (higher \( E[\text{CRB}(\phi)] \)) while maintaining ISAC sensing performance (lower \( E[\text{CRB}(\theta)] \)). From Fig.\ref{regions}\subref{subpromlemcrb}, it is observed that $E[\text{CRB}(\theta)]$ for SLB at subproblem 3 (secure sensing optimal) is lower than that for SSJB at subproblem 1 (sensing optimal without security), indicating improved sensing performance at the base station. However, the sensing eav’s ergodic rate in SLB at subproblem 3 is lower than in SSJB at subproblem 1, indicating a degradation in target privacy. This trade-off arises because neither of these subproblems is globally optimal for the entire system, they are optimized for specific aspects (i.e., sensing performance or security), rather than for joint system performance. Moreover, we note that in Fig.\ref{regions}\subref{subpromlemcrb}, for SLB in the sub2 case, the ergodic CRB of the sensing eavesdropper (BS) increases, whereas in Fig.\ref{crbfigures}\subref{secondbssensing}, it was decreasing (increasing). This discrepancy arises because, in Fig.\ref{crbfigures}\subref{secondbssensing}, only one parameter is varied across the curves while the other two are held constant. However, in Fig.\ref{regions}\subref{subpromlemcrb}, although $\tau_3$ is fixed in the SLB sub2 case, an increase in $\tau_1$ is accompanied by a decrease in $\tau_2$. As a result, the ergodic CRB at the sensing eavesdropper (BS) decreases (increases) due to the increase in $\tau_1$, and increases (decreases) due to the decrease in $\tau_2$. The net effect is that the ergodic CRB at the sensing eavesdropper (BS) increases overall in Fig.\ref{regions}\subref{subpromlemcrb}.
\section{Conclusion}\label{conclusion}
We analyze security and privacy in a random downlink MIMO ISAC system. A multi-antenna BS employs two different precoding schemes, namely SSJB and SLB, to simultaneously:  
1) Transmit data to a communication user,  
2) Sense a target's angle (which may be malicious), and  
3) Counter eavesdropping attempts. The system considers both external communication eaves and sensing eaves. For both SLB and SSJB schemes, we derive: ESR, the ergodic CRB, and CRB outage probability for both the BS and sensing eaves, taking into account channel randomness. Our analysis demonstrates three fundamental tradeoff in random ISAC networks: 1) Sens-
ing accuracy vs. communication rate as other traditional ISAC
works with no security and privacy constraints; 2) Sensing
accuracy vs. secure communication rate; 3) Communication security vs. sensing privacy. 
\begin{figure*}
    \centering
     \subfigure[Ergodic CRB($\theta)$ at BS/ESR\label{region1}]{\includegraphics[scale=.42]{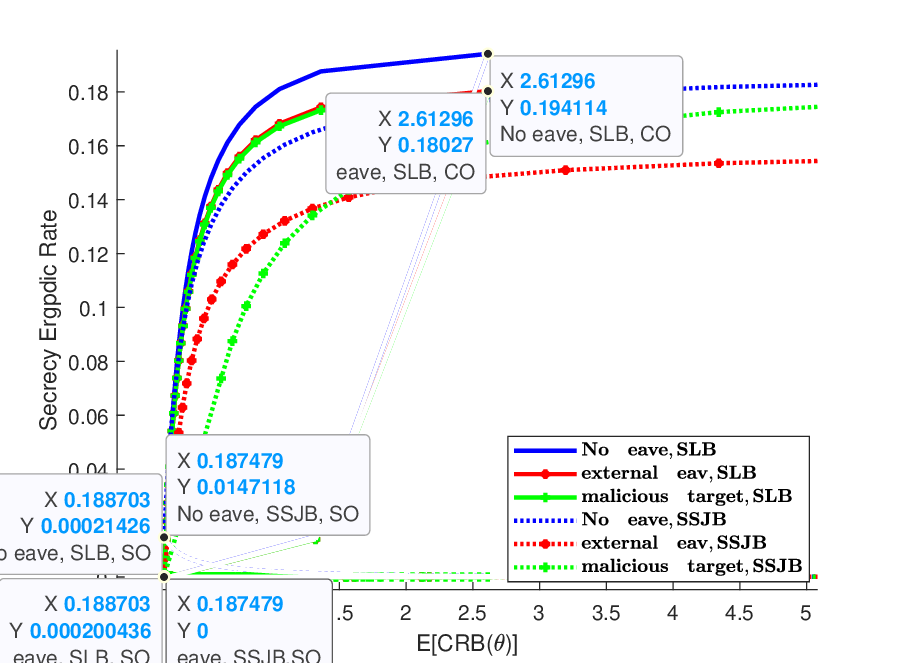}}\hspace{-0.5cm}
    \subfigure[Ergodic CRB($\phi)$ at strong sensing eav/ESR \label{region2}]{\includegraphics[scale=.42]{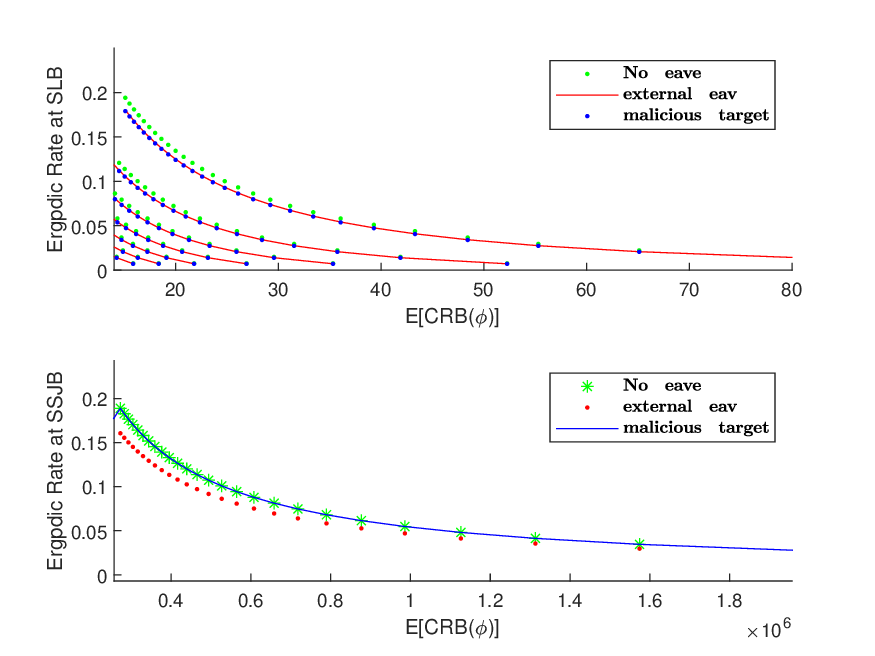}}\hspace{-0.5cm}
    \subfigure[Ergodic CRB($\phi)$ (strong)/ CRB($\theta)$ at subproblems\label{subpromlemcrb}]
    {\includegraphics[scale=.42]{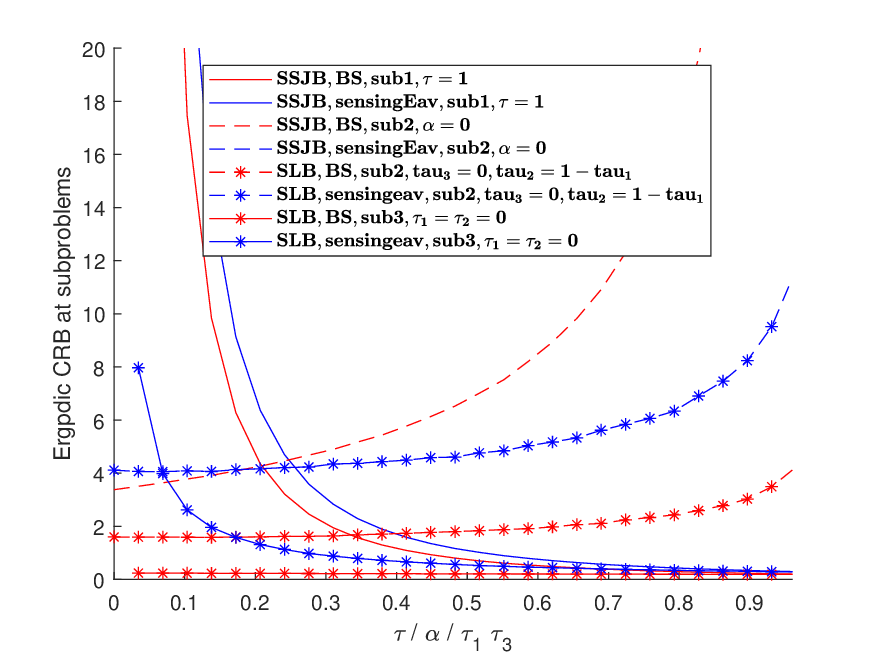}}\hspace{-0.5cm}
    \caption{trade off regions of ESR versus ergodic CRB}
    \label{regions}
\end{figure*}
\bibliographystyle{IEEEtran}
\bibliography{ref}
\appendices
\section{Proof of Lemma \ref{privacy}}\label{privacyproof}
To solve the optimization problem (\ref{optimization}) where $\text{CRB}(\theta)$ is given at (\ref{crbnew1}) and $\text{CRB}(\phi)$ is derived at subsection \ref{opoftargeteav}, we follow the same approach as \cite[Appendix C]{rangecompressionandwaveformoptimization}. Let $\mathbf{U}=[\frac{\mathbf{a}}{||\mathbf{a}||}, \frac{\mathbf{a}'}{||\mathbf{a}'||}]$ and $\mathbf{R}_{x} = \boldsymbol{\Delta} \boldsymbol{\Delta}^H.$ Decompose $\boldsymbol{\Delta}$ as $\boldsymbol{\Delta} = \mathbf{P}_U \boldsymbol{\Delta} + \mathbf{P}_U^{\perp} \boldsymbol{\Delta}$, where $\mathbf{P}_U$ denotes the orthogonal projection onto the subspace spanned by the columns of $\mathbf{U}$, and $\mathbf{P}_U^{\perp} = \mathbf{I} - \mathbf{P}_U$, with $\mathbf{I}$ denoting the identity matrix. Therefore, we have $\mathbf{R}_{x}= \mathbf{P}_U \boldsymbol{\Delta} \boldsymbol{\Delta}^H \mathbf{P}_U + \hat{\mathbf{R}}_{x}$
where $\hat{\mathbf{R}}_{x} = \mathbf{P}_U^{\perp} \boldsymbol{\Delta} \boldsymbol{\Delta}^H \mathbf{P}_U^{\perp} + \mathbf{P}_U \boldsymbol{\Delta} \boldsymbol{\Delta}^H \mathbf{P}_U^{\perp} + \mathbf{P}_U^{\perp} \boldsymbol{\Delta} \boldsymbol{\Delta}^H \mathbf{P}_U.$ Moreover, due to the definition of $\mathbf{P}_U $, we have: $ \mathbf{a}^H(\theta) \hat{\mathbf{R}}_{x} \mathbf{a}(\theta) = {\mathbf{a}'}^H(\theta) \hat{\mathbf{R}}_{x} {\mathbf{a}'}(\theta) = {\mathbf{a}'}^H(\theta) \hat{\mathbf{R}}_{x} \mathbf{a}(\theta) ={\mathbf{a}}^H(\theta) \hat{\mathbf{R}}_{x} \mathbf{a}'(\theta) =0.$ Thus, $\hat{\mathbf{R}}_{x}$ contributes neither to $\text{CRB}(\theta)$ nor to $\text{CRB}(\phi)$, since $\mathbf{R}_x$ appears in $\text{CRB}(\theta)$ and $\text{CRB}(\phi)$ only through the terms $\mathbf{a}^H(\theta) \mathbf{R}_x \mathbf{a}(\theta)$, ${\mathbf{a}'}^H(\theta) \mathbf{R}_x {\mathbf{a}'}(\theta)$, ${\mathbf{a}'}^H(\theta) \mathbf{R}_x \mathbf{a}(\theta)$, and $\mathbf{a}^H(\theta) \mathbf{R}_x \mathbf{a}'(\theta)$. Therefore, $\hat{\mathbf{R}}_x$ does not contribute to the objective function in (\ref{optimization}). Next, note that $\mathrm{tr}(\hat{\mathbf{R}}_{x}) = \mathrm{tr}\left( \mathbf{P}_U^{\perp} \boldsymbol{\Delta} \boldsymbol{\Delta}^H \mathbf{P}_U^{\perp} + \mathbf{P}_U \boldsymbol{\Delta} \boldsymbol{\Delta}^H \mathbf{P}_U^{\perp} + \mathbf{P}_U^{\perp} \boldsymbol{\Delta} \boldsymbol{\Delta}^H \mathbf{P}_U \right) = \mathrm{tr}\left( \mathbf{P}_U^{\perp} \boldsymbol{\Delta} \boldsymbol{\Delta}^H \mathbf{P}_U^{\perp} \right) = \left\| \boldsymbol{\Delta}^H \mathbf{P}_U^{\perp} \right\|_F^2 \geq 0$, where $\|\cdot\|_F$ denotes the Frobenius matrix norm. The equality holds if and only if $\boldsymbol{\Delta}^H \mathbf{P}_U^{\perp} = 0$, which implies $\hat{\mathbf{R}}_x = 0$. Hence, we have proved that, while the objective function does not depend on $(\hat{\mathbf{R}}_{x})$, any $\hat{\mathbf{R}}_{x} \neq 0$ will increase $\mathrm{tr}(\mathbf{R}_{x})$ compared with the case of $\hat{\mathbf{R}}_{x} = 0$. The conclusion is that under the total power constraint, i.e., $\mathrm{tr}(\mathbf{R}_{x}) =p_t$, we necessarily must have $\mathrm{tr}(\hat{\mathbf{R}}_{x}) = 0$, i.e., $\hat{\mathbf{R}}_{x} = 0$. Hence, $
\mathbf{R}_{x} = \mathbf{P}_U \boldsymbol{\Delta} \boldsymbol{\Delta}^H \mathbf{P}_U \triangleq \mathbf{U} \boldsymbol{\Lambda} \mathbf{U}^H,
$ where $\boldsymbol{\Lambda}$ is a positive semi-definite $2 \times 2$ matrix. Thus, $\mathbf{X}$ is a linear combination of the columns of $\mathbf{U}$, i.e., $\frac{\mathbf{a}}{||\mathbf{a}||}$ and $\frac{\mathbf{a}'}{||\mathbf{a}'||}$.
\section{Proof of Lemma \ref{exactderivationlemma}}\label{proofexactderivationlemma}
We first prove the simplified expression of $\text{CRB}(\theta)$, and then provide the proofs for the lower bound and the approximation of $P(\text{CRB}(\theta) > \epsilon)$, respectively. By substituting (\ref{Rx}) into (\ref{crbnew1}), and noting that $\mathbf{a}^H\mathbf{t}_1\mathbf{t}_1^H\mathbf{a}'$ is a scalar, along with the orthogonality between $\mathbf{a}$ and $\mathbf{a}'$ and the condition $\mathbf{a}^H\mathbf{t}_i = 0$ for $i = 3, \dots, N$, we obtain—after some algebraic manipulation—the following expression: $\text{CRB}(\theta) = \frac{Q}{\gamma_1 |\mathbf{b}'|^2 |\mathbf{a}^H \mathbf{t}1|^2 + \gamma_2 |\mathbf{b}|^2 \sum{i=3}^{N} |\mathbf{a}'^H \mathbf{t}_i|^2}.$ Moreover by replacing $\mathbf{t}_1= \alpha \tilde{\mathbf{a}} + \beta \tilde{\mathbf{h}}$, we have: $|\mathbf{a}^H\mathbf{t}_1|^2=|\alpha|^2N$ where we have used the orthogonality between $\mathbf{a}$ and $\tilde{\mathbf{h}}$ and $|\mathbf{a}|^2=N$. furthermore, we have:$\sum_{i=3}^{N}|\mathbf{a}'^H\mathbf{t}_i|^2=\mathbf{a}'^H|(\mathbf{I}-\tilde{\mathbf{a}}\tilde{\mathbf{a}}^H-\tilde{\mathbf{h}}\tilde{\mathbf{h}}^H)|\mathbf{a}$ since the matrix $\begin{bmatrix} \mathbf{\tilde{a}} & \mathbf{\tilde{h}} & \mathbf{G} \end{bmatrix},$ where \(\mathbf{G}=\begin{bmatrix} \mathbf{t}_3 & ... & \mathbf{t}_N \end{bmatrix}\), constructs an orthonormal basis for the \(N\)-dimensional space and is unitary. after some mathematical operation and replacing the definition of $\mathbf{\tilde{a}}$ and $\mathbf{\tilde{h}}$ and noting that $\mathbf{a}'(\theta)^H\mathbf{a}(\theta)=0$, we have $\sum_{i=3}^{N}|\mathbf{a}'^H\mathbf{t}_i|^2=||\mathbf{a}'||^2-\frac{|\mathbf{a}'^H\mathbf{h}|^2}{||\mathbf{h}||^2-\frac{1}{N}|\mathbf{a}^H\mathbf{h}|^2}$. Thus, by noting that $||\mathbf{b}||^2=M$ we have 
$\text{CRB}(\theta)=\frac{Q}{\gamma_1||\mathbf{b}'||^2N|\alpha|^2+\gamma_2M(||\mathbf{a}'||^2-\frac{|\mathbf{a}'^H\mathbf{h}|^2}{||\mathbf{h}||^2-\frac{1}{N}|\mathbf{a}^H\mathbf{h}|^2}))}$. By rewriting $\mathbf{a}$ and $\mathbf{h}$ element wise and after some operation the proof is complete. 

\textbf{Proof of the lower bound:} If we ignore the term $\big((\sum_{i=1}^{N}-f'_i\hat{t}_i)^2+(\sum_{i=1}^{N}f'_i\hat{r}_i)^2\big)$ at the denominator of CRB($\theta$), the the total CRB will be smaller. This is due to Cauchy–Schwarz inequality for $N$-dimensional complex space, we have: $|\sum_{i=1}^{N}\!e^{jf_i}{h}_i|^2=R^2+T^2<(\sum_{i=1}^{N}\!|e^{jf_i}|^2)(\sum_{i=1}^{N}\!|{h}_i|^2)\nonumber=N(\sum_{i=1}^{N}\!|h_i|^2)=NK,$ thus the term $\frac{(\sum_{i=1}^{N}-f'_i{t}_i)^2+(\sum_{i=1}^{N}f'_i{r}_i)^2}{K-\frac{1}{N}(T^2+R^2)}$ is positive and by ignore the term $\big((\sum_{i=1}^{N}-f'_i\hat{t}_i)^2+(\sum_{i=1}^{N}f'_i\hat{r}_i)^2\big)$ the the denominator of CRB($\theta$) will be bigger. Thus, we obtain a lower bound on CRB, named as LCRB, which is $\frac{6\sigma^2_R}{\cos^2(\theta)MN \pi ^2LP \mid c_3 \mid ^2\big(|\alpha|^2\tau(M^2-1)+\frac{(N^2-1)(1-\tau)}{N-2}\big)}$. Thus, $P_c(\epsilon)=P(\text{CRB}(\theta)>\epsilon)>P(\text{LCRB}(\theta)\!\!>\epsilon)\triangleq P_{Lc}(\epsilon)$ and $P_{Lc}$ provides a lower bound on the OP of the target.
Since $\theta$ is uniformly distributed in the interval $[-\pi/2,\pi/2]$, after calculating the CDF of $\cos^2(\theta)$, the proof is complete.

\textbf{Proof of the approximation:} When $N$ is large, due to the law of large numbers, we have:
$(\sum_{i=1}^{N}-f'_i{t}_i)^2+(\sum_{i=1}^{N}f'_i{r}_i)^2\overset{(p)}{\rightarrow} E\{(\sum_{i=1}^{N}-f'_i{t}_i)^2+(\sum_{i=1}^{N}f'_i{r}_i)^2 \}.$ Thus, by replacing $(\sum_{i=1}^{N}-f'_i{t}_i)^2+(\sum_{i=1}^{N}f'_i{r}_i)^2\big)$ in the denominator of CRB($\theta$) with $E\{(\sum_{i=1}^{N}-f'_i{r}_i)^2+(\sum_{i=1}^{N}f'_i{t}_i)^2 \}\overset{(a)}=\sum_{i=1}^{N}(f'_i)^2E\{{{k_i}}\}=\sum_{i=1}^{N}(f'_i)^2$, where (a) is due to independence of ${t}_i$ with ${t}_j$ and ${r}_i$ with ${r}_j$, we will have an approximation of CRB($\theta$) named ACRB.
Thus, we have: $P_c(\epsilon)\!=\!P(\text{CRB}(\theta)\!>\!\epsilon)\!\approx\!\! P(\text{ACRB}(\theta)\!>\!\epsilon)\!\!\triangleq\! P_{Ac}(\epsilon)$. Since \( \theta \) is uniformly distributed in the interval \([-\pi/2,\pi/2]\), $||\mathbf{a}'||^2=(\sum_{i=1}^{N}\!|f'_i|^2)$, and conditioning on \( \theta \) while noting that \( \theta \) and \( \mathbf{h} \) are independent, we obtain: $P_{Ac}(\epsilon)\!\!=\!\!\!\int_{-\pi/2}^{\pi/2}\!\!\!\!\!\!P(\frac{\frac{6\sigma_R^2}{L|c_3|^2\pi^2\cos^2(\theta)MN}}{\gamma_1|\alpha|^2(M^2-1)+\gamma_2(N^2-1)(1-\frac{1}{K-\frac{1}{N}(T^2+R^2)}))}>\epsilon)|\theta)f_{\theta}(\theta)d\theta$, where given \( \theta \), the only random variables are \( R \), \( T \), and \( K \). Thus, to compute the inner probability, we need to derive the joint probability density function (PDF) of \( R \), \( T \), and \( K \). These variables (as well as \( r_i \), \( t_i \), and \( k_i \)) are not independent since they are functions of \( |h_i| \) (which follows a Rayleigh distribution), \( \tilde{\phi}_i \) (which is uniformly distributed over \( [0, 2\pi) \)), and \( f_i \) (a function of \( \theta \)).  
By conditioning on \( \theta \), \( f_i \), for all \( i=1,...,N \), is treated as constant in the following derivations. We define \( N \) random vectors, \( \mathbf{d}_i=[r_i, t_i, k_i]^T \in \mathbb{R}^{3 \times 1} \), for \( i=1,...,N \). For any pair of \( j \) and \( i\neq j \), the vectors \( \mathbf{d}_j \) and \( \mathbf{d}_i \) are i.i.d. because \( h_i \)'s and \( \tilde{\phi}_i \)'s are i.i.d.  
Thus, by applying the multidimensional central limit theorem (CLT)\cite{enwiki:1155685628}, when \( N \) is large \footnote{Section \ref{simulations} shows that for \( N>9 \), the multidimensional CLT holds. Additionally, in \cite{Aunifiedperformanceframeworkfor}, the authors demonstrate via simulations that CLT accuracy for a one-dimensional random variable holds for \( N>8 \).}, we have: $\sqrt{N}[\frac{1}{N}(\sum_{i=1}^{N}\mathbf{d}_i)-\mathbf{\mu_d}]\overset{d}{\rightarrow} \mathcal{N}_3(\mathbf{0},\mathbf{\Sigma_d}),$ which implies $\sum_{i=1}^{N}\mathbf{d}_i\overset{d} \rightarrow \mathcal{N}_3(N\mathbf{\mu_d}, N\mathbf{\Sigma_d}),$ where \( \mathbf{\mu_d} \) and \( \mathbf{\Sigma_d} \) are the mean vector and covariance matrix of \( \mathbf{d}_i \), respectively (identical for all \( i=1,...,N \)). Therefore, $[R, T, K]^T\overset{(d)}{\rightarrow} \mathcal{N}_3(N\mathbf{\mu_d},N\mathbf{\Sigma_d}).$\footnote{Monte Carlo simulations confirm the numerical approximation.}  
By determining \( \mathbf{\mu_d} \) and \( \mathbf{\Sigma_d} \) using Lemma \ref{lemma1i}, we derive the joint PDF of \( R \), \( T \), and \( K \).  
\begin{lemma} \label{lemma1i}
\(\mathbf{\mu_d}\) and \(\mathbf{\Sigma_d}\) are given by: $\mathbf{\mu}_d=[0, 0, 1]^T$ and  
\begin{align}  
\mathbf{\Sigma}_d&=\begin{bmatrix}  
\frac{1}{2}& 0 & 0 \\  
0& \frac{1}{2} & 0 \\  
0& 0 & 1   
\end{bmatrix}.\label{rtk}  
\end{align}
\end{lemma}
\textbf{Proof:} First, we derive ${{\mathbf{\mu_d}}} = [E[{{r}}_i], E[{{t}}_i], E[{{k}}_i]]$. Since \( |h_i|^2 \) follows an exponential distribution with a mean of 1, we have $\mathbb{E}[|h_i|^2] = 1$. The terms in $r_i$ and $k_i$ contain cosine functions involving the random variables \( \tilde{\phi}_i \). The expectations of these cosine terms are zero because the angles are uniformly distributed. Thus, we obtain ${{\mathbf{\mu_d}}} = [0, 0, 1]$.  
Next, to calculate each element of the covariance matrix $\mathbf{{{\Sigma_d}}}$, we use the independence of $|h_i|$ and $\theta$, along with $E(|h_i|^4) = 2$. Additionally, we utilize the following identities: $\cos(A)\cos(B) = \frac{1}{2} [\cos(A - B) + \cos(A + B)], \cos(A)\sin(A) = \frac{1}{2} \sin(2A), \sin^2(A) = \frac{1}{2} (1 - \cos(2A)).$
Due to space limitations, we omit the detailed derivations.
\\
From Lemma \ref{lemma1i}, we see that the joint conditional PDF of \( R \), \( T \), and \( K \) is independent of \( \theta \) and \( \phi \). In fact, \( R \), \( T \), and \( K \) are independent Gaussian random variables since the off-diagonal elements of the covariance matrix are zero. Thus, their joint PDF is the product of three Gaussian distributions, each with its corresponding mean and variance. Using this property, along with the assumption that \( \theta \) is uniformly distributed, the proof is complete.
\section{Proof of Lemma \ref{exactderivationlemma2}}\label{proofexactderivationlemma2}
We first prove the simplified expression of $\text{CRB}(\theta)$, and then provide the proofs for the upper bound, lower bound, and the approximation of $P(\text{CRB}(\theta) > \epsilon)$, respectively. By using (\ref{x2}), along with the orthogonality between the radar signals, data, and artificial noise (AN), and noting that $\begin{bmatrix} \frac{\mathbf{h}}{\|\mathbf{h}\|} & \tilde{\mathbf{G}} \end{bmatrix}$ is a unitary matrix that forms an orthonormal basis for the $N$-dimensional space, we have:${\mathbf{R}}_x = \frac{1}{L} {\mathbf{X}} {\mathbf{X}}^H=\frac{P\tau_1}{||\mathbf{h}||^2}\mathbf{h}\mathbf{h}^H+ \frac{P\tau_2}{N-1}(\mathbf{I}-\frac{\mathbf{h}\mathbf{h}^H}{|\mathbf{h}|^2})+\frac{P\tau_3}{N}\mathbf{a}\mathbf{a}^H+\frac{P(1-\tau_1-\tau_2-\tau_3)}{|\mathbf{a}'|^2}\mathbf{a}'\mathbf{a}'^H$. Moreover, due to the orthogonality between $\mathbf{a}$ and $\mathbf{a}'$, we have: $\mathbf{a}^H{\mathbf{R}}_x\mathbf{a}=\frac{(P\tau_1-\frac{P\tau_2}{N-1})}{||\mathbf{h}||^2}|\mathbf{a}^H\mathbf{h}|^2+NP\tau_3+\frac{NP\tau_2}{N-1}$, $\mathbf{a}'^H{\mathbf{R}}_x\mathbf{a}'=\frac{(P\tau_1-\frac{P\tau_2}{N-1})}{||\mathbf{h}||^2}|\mathbf{a}'^H\mathbf{h}|^2+\frac{|\mathbf{a}'|^2P\tau_2}{N-1}+P(1-\tau_1-\tau_2-\tau_3)|\mathbf{a}'|^2$, and $\mathbf{a}^H{\mathbf{R}}_X\mathbf{a}'=\frac{(P\tau_1-\frac{P\tau_2}{N-1})}{||\mathbf{h}||^2}\mathbf{a}^H\mathbf{h}\mathbf{h}^H\mathbf{a}'$. Thus, after defining $y\triangleq (P\tau_1-\frac{P\tau_2}{N-1})$, $x\triangleq NP\tau_3+\frac{NP\tau_2}{N-1}$, and after some mathematical operation (\ref{crbnew1}) becomes:
$\text{CRB}(\theta)=Q[||\mathbf{b}'||^2(\frac{y}{||\mathbf{h}||^2}|\mathbf{a}^H\mathbf{h}|^2+x)+||\mathbf{b}||^2(\frac{|\mathbf{a}'|^2P\tau_2}{N-1}+P(1-\tau_1-\tau_2-\tau_3)|\mathbf{a}'|^2+\frac{x\frac{y}{||\mathbf{h}||^2}|\mathbf{a}'^H\mathbf{h}|^2}{(\frac{y}{||\mathbf{h}||^2}|\mathbf{a}^H\mathbf{h}|^2+x)})]^{-1}$. Moreover, by rewriting \(\mathbf{a}\) and \(\mathbf{h}\) element-wise and after performing some operations, the \(\text{CRB}(\theta)\) is derived as \ref{crbsimplified} and the proof is complete.

\textbf{Proof of the upper bound} If we ignore the term \(\big((\sum_{i=1}^{N}-f'_i\hat{t}_i)^2+(\sum_{i=1}^{N}f'_i\hat{r}_i)^2\big)\) in the denominator of CRB\((\theta)\) in (\ref{crbsimplified}), the total CRB will be larger by the assumption that $y$ is positive. Thus, we obtain an upper bound on the CRB, and by following the same approach as in Appendix \ref{proofexactderivationlemma}, an upper bound for $P(\text{CRB}(\theta) > \epsilon)$ is derived

\textbf{Proof of the lower bound} Using the Cauchy–Schwarz inequality for an \(N\)-dimensional complex space, we have:$(\sum_{i=1}^{N}-f'_i{t}_i)^2+(\sum_{i=1}^{N}f'_i{r}_i)^2=|\sum_{i=1}^{N}\!jf'_ie^{jf_i}{h}_i|^2  
<(\sum_{i=1}^{N}\!|jf'_i|^2)(\sum_{i=1}^{N}\!|e^{jf_i}{h}_i|^2)=(\sum_{i=1}^{N}\!|f'_i|^2)({{K}}),$ where \(\sum_{i=1}^{N}\!|f'_i|^2=\frac{\pi^2\cos^2(\theta)N(N^2-1)}{12}\). Replacing $(\sum_{i=1}^{N} -f'_i t_i)^2 + (\sum_{i=1}^{N} f'_i r_i)^2$ with $K \sum_{i=1}^{N} |f'_i|^2$ in the denominator of CRB$(\theta)$ in (\ref{crbsimplified}) yields a lower bound. The proof follows the same approach as in Appendix \ref{proofexactderivationlemma}. Similarly, the approximation of $P(\text{CRB}(\theta) > \epsilon)$ is derived analogously.
\section{Proof of Lemma \ref{crbsensingeavweak}}\label{proofcrbsensingeavweak}
When eav is not strong, it doesn't know \(\mathbf{s}_u\) and \(\mathbf{V}\), thus using $\mathbf{X}$ at (\ref{x}) we have $\mathbb{E}[\mathbf{X}] = \sqrt{P\tau} \mathbf{t}_1 \mathbb{E}[\mathbf{s}_u] + \sqrt{P(1-\tau)} \sum_{i=3}^{N} \mathbf{t}_i \mathbb{E}[\mathbf{v}_i]=0$. So, based on the model in the subsection \ref{opoftargeteav}, the revived signal at the eav is $\tilde{\mathbf{y}}_{sr} = \tilde{\mathbf{u}_e} + \tilde{\mathbf{Z}_{sr}}, \quad \tilde{\mathbf{u}}_e = c_4 \text{vec}(\mathbf{B}\mathbf{X}),$ where $\mathbb{E}[\tilde{\mathbf{u}}_e] = c_4 \text{vec}(\mathbf{B} \cdot 0) = 0.$ which results 
$\mathbb{E}[\tilde{\mathbf{y}}_{sr}] = 0$. Moreover,
\[
\text{Var}(\tilde{\mathbf{y}}_{sr}) = \mathbb{E}[\tilde{\mathbf{y}}_{sr}\tilde{\mathbf{y}}_{sr}^H]=\mathbb{E}[\tilde{\mathbf{u}}_e\tilde{\mathbf{u}}^H_e] + \mathbb{E}[\tilde{\mathbf{Z}_{sr}}\tilde{\mathbf{Z}_{sr}}^H].
\] Also, we have :
$\mathbb{E}[\tilde{\mathbf{u}}_e\tilde{\mathbf{u}}^H_e] = |c_4|^2 \mathbb{E}[\text{vec}(\mathbf{B}\mathbf{X})\text{vec}(\mathbf{B}\mathbf{X})^H]=(\mathbf{I}_L \otimes \mathbf{B}) \mathbb{E}[\text{vec}(\mathbf{X})\text{vec}(\mathbf{X})^H] (\mathbf{I}_L \otimes \mathbf{B}^H)$ where we have used the property \(\text{vec}(\mathbf{B}\mathbf{X}) = (\mathbf{I}_L \otimes \mathbf{B}) \text{vec}(\mathbf{X})\). after some mathematical operation we get:
\begin{align}
&\text{Var}(\tilde{\mathbf{y}}_{sr}) = |c_4|^2 ( P\tau L \mathbf{I}_L \otimes (\mathbf{B} \mathbf{t}_1 \mathbf{t}_1^H \mathbf{B}^H) + P(1-\tau) \frac{L}{N-2}\nonumber\\
&\sum_{i=3}^{N} \mathbf{I}_L \otimes (\mathbf{B} \mathbf{t}_i \mathbf{t}_i^H \mathbf{B}^H) ) + \sigma^2_r \mathbf{I}_{ML}.
\end{align}
By noting the orthogonality between $\mathbf{a}$ and $\mathbf{t}_i$ for $i = 3, \ldots, N$, and the definition of $\mathbf{B}$ and $\mathbf{t}_i$ we get:$ \text{Var}(\tilde{\mathbf{y}}_{sr}) = \mathbf{C} = |c_4|^2 L \mathbf{I}_L \otimes (d \mathbf{c}\mathbf{c}^H) + \sigma^2_r \mathbf{I}_{ML},$ where $d=P\tau |\alpha|^2|\mathbf{a}|^2$. Moreover,
$\mathbf{F}_{\phi\phi} = \text{Tr}\left( \mathbf{C}^{-1} \frac{\partial \mathbf{C}}{\partial \phi} \mathbf{C}^{-1} \frac{\partial \mathbf{C}}{\partial \phi} \right)$, where 
\(\frac{\partial \mathbf{C}}{\partial \phi} = |c_4|^2 L d \mathbf{I}_L \otimes (\mathbf{c}' \mathbf{c}^H + \mathbf{c} \mathbf{c}'^H)\). Also, by using the Woodbury identity, the inverse of \(\mathbf{C}\) is:
\[
\mathbf{C}^{-1} = \frac{1}{\sigma_r^2} \mathbf{I}_{ML} - \frac{|c_4|^2 L d}{\sigma_r^2 (\sigma_r^2 + |c_4|^2 L d N_e)} \mathbf{I}_L \otimes \mathbf{c}\mathbf{c}^H.
\]Thus using \(\mathbf{c}^H \mathbf{c}' = 0\) and \(\mathbf{c}^H \mathbf{c} = N_e\):
\begin{align}
&\left( \mathbf{C}^{-1} \frac{\partial \mathbf{C}}{\partial \phi} \right)^2 =\frac{|c_4|^4 L^2 d^2}{\sigma_r^4} \mathbf{I}_L \otimes (N_e\mathbf{c}'\mathbf{c}'^H + \|\mathbf{c}'\|^2\mathbf{c}\mathbf{c}^H)\nonumber\\
&-\frac{2|c_4|^6 L^3 d^3 N_e^2}{\sigma_r^4(\sigma_r^2 + |c_4|^2 L d N_e)} \mathbf{I}_L \otimes \mathbf{c}'\mathbf{c}'^H.
\end{align}
Since the trace of a Kronecker product is: $\text{Tr}(\mathbf{I}_L \otimes \mathbf{A}) = L\text{Tr}(\mathbf{A})$, after some manipulation we have: $\mathbf{F}_{\theta\theta} = \frac{2|c_4|^4 L^3 d^2 N_e \|\mathbf{c}'\|^2}{\sigma_r^4} - \frac{2|c_4|^6 L^4 d^3 N^2_e \|\mathbf{c}'\|^2}{\sigma_r^4(\sigma_r^2 + |c_4|^2 L d N_e)} = \frac{2|c_4|^4 L^3 d^2 N_e \|\mathbf{c}'\|^2}{\sigma_r^2(\sigma_r^2 + |c_4|^2 L d N_e)}$.

Also, $\mathbf{F}_{\phi\bar{\alpha}_e} = \text{Tr}\left( \mathbf{C}^{-1} \frac{\partial \mathbf{C}}{\partial \phi} \mathbf{C}^{-1} \frac{\partial \mathbf{C}}{\partial \bar{\alpha}_e} \right)$ where \(\frac{\partial \mathbf{C}}{\partial \bar{\alpha}_e} \propto \mathbf{I}_L \otimes \mathbf{c}\mathbf{c}^H\) (since \(\mathbf{C}\) depends on \(|c_4|^2\)). By following the same approach as the derivation of $\mathbf{F}_{\phi\phi}$ and due to the orthogonality between $\mathbf{c}$ and $\mathbf{c'}$ we have:
$\mathbf{F}_{\phi\bar{\alpha}_e} = 0$. Thus using (\ref{crbphi}) and substituting  \(\|\mathbf{c}'\|^2 = \pi^2\cos^2(\phi)\frac{N_e(N^2_e-1)}{12}\):we have:
\begin{align}
\text{CRB}(\phi) = \frac{1}{\mathbf{F}_{\phi\phi}} =  \frac{6\sigma_r^2(\sigma_r^2 + |c_4|^2 L d N_e)}{|c_4|^4 L^3 d^2 \pi^2 \cos^2(\phi) N^2_e (N^2_e - 1)}\label{crbforweak}
\end{align}
 where $d=P\tau |\alpha|^2N$. Since $\phi \sim \mathcal{U}(-\pi/2, \pi/2)$, the proof is completed by calculating the CDF of $\cos^2(\phi)$.
 \section{Proof of Lemma \ref{crbsensingeav2weak}}\label{proofcrbsensingeav2weak}
By following the same approach as in Appendix \ref{proofcrbsensingeavweak}, using $\mathbf{X}$ at (\ref{x2}), and $\mathbf{{R}}_x$ at Appendix \ref{proofexactderivationlemma2}, and after some mathematical operation, we have $\mathbb{E}[\tilde{\mathbf{y}}_{sr}] = 0$ and:
\begin{align}
&\text{Var}(\tilde{\mathbf{y}}_{sr}) = |c_4|^2 L P [ ( \tau_1 - \frac{\tau_2}{N-1} ) \frac{|\mathbf{a}^H \mathbf{h}|^2}{\|\mathbf{h}\|^2} + \frac{\tau_2 N}{N-1} + \tau_3 N ]\nonumber\\
&\mathbf{c}\mathbf{c}^H \otimes \mathbf{I}_L + \sigma_r^2 \mathbf{I}_{ML}
\end{align}
Thus, it can be write as $ \text{Var}(\tilde{\mathbf{y}}_{sr}) = \mathbf{C} = |c_4|^2 L \mathbf{I}_L \otimes (\tilde{d} \mathbf{c}\mathbf{c}^H) + \sigma^2_r \mathbf{I}_{ML},$ where $\tilde{d}=P [ ( \tau_1 - \frac{\tau_2}{N-1} ) \frac{|\mathbf{a}^H \mathbf{h}|^2}{\|\mathbf{h}\|^2} + \frac{\tau_2 N}{N-1} + \tau_3 N ]$. thus, based on the Appendix \ref{proofcrbsensingeavweak}, (\ref{crbforweak}) is valid by substituting $\tilde{d}$ with $d$ in . Thus, we have:
\begin{align}
\text{CRB}(\phi) = \frac{1}{\mathbf{F}_{\phi\phi}} =  \frac{6\sigma_r^2(\sigma_r^2 + |c_4|^2 L \tilde{d} N_e)}{|c_4|^4 L^3 \tilde{d}^2 \pi^2 \cos^2(\phi) N^2_e (N^2_e - 1)}
\end{align}
Thus by rewriting \(\mathbf{a}\) and \(\mathbf{h}\) element-wise and following the same approach as in Appendix \ref{proofexactderivationlemma}, the proof is complete.
\section{Proof of Lemma \ref{newSLB}}\label{proofnewSLB}
By truncating the domain to $\theta \in [-\frac{\pi}{2} + \delta, \frac{\pi}{2} - \delta]$ for a small $\delta > 0$, and using the lower bound of $\text{CRB}(\theta)$—denoted as LCRB—derived in Appendix \ref{proofexactderivationlemma}, and noting that the only random variable in its expression is $\theta$, along with the identity $\mathbb{E}_\theta \left[ \frac{1}{\cos^2(\theta)} \right] = \frac{2 \cot(\delta)}{\pi - 2\delta},$ we derive the expected value $\mathbb{E}[\text{LCRB}(\theta)]$. Following the same approach, and using the derived expressions for $\text{CRB}(\phi)$ for the strong and weak eav in Lemmas \ref{crbsensingeav} and \ref{crbsensingeavweak}, respectively, the ergodic CRB for both the strong and weak eav can be obtained. Moreover, according to Lemma \ref{lemma1i}, for large $N$, the random variables behave as follows:
$R \sim \mathcal{N}(0, N/2)$, $T \sim \mathcal{N}(0, N/2)$. Hence, $T^2 + R^2 \sim N \cdot \text{Exp}(1)$, an exponential distribution with mean $N$. Also, let $K \sim \mathcal{N}(N, N)$ for $K > 0$. Thus, $f_K(k) = \frac{1}{\sqrt{2 \pi N}} e^{-\frac{(k - N)^2}{2N}} \cdot \frac{1}{1 - \Phi\left(-\frac{N}{\sqrt{N}}\right)}, \quad k > 0,$ where $\Phi$ is the cumulative distribution function (CDF) of the standard normal distribution $\mathcal{N}(0, 1)$. For $N \geq 1$, the truncation effect is negligible, and $\mathbb{E}[K] \approx N$. Moreover, by Lemma \ref{lemma1i}, the random variables $R$, $T$, and $K$ are independent, as their joint covariance matrix has zero correlation. Consequently, since $S = R^2 + T^2$ is a function of $R$ and $T$, it is also independent of $K$ (a function of random variables that are independent of another random variable remains independent of that variable).
Therefore, the joint PDF of $S$ and $K$ is the product of their marginals:
$f_{S, K}(s, k) = f_S(s) \cdot f_K(k) = \frac{1}{N \sqrt{2 \pi N}} e^{-\frac{s}{N} - \frac{(k - N)^2}{2N}}, \quad s \geq 0, \, k > 0.$ Thus, using the approximation of $\text{CRB}(\theta)$—denoted as ACRB—derived in Appendix \ref{proofexactderivationlemma}, and noting that
$\mathbb{E}_\theta \left[ \frac{1}{\cos^2(\theta)} \right] = \frac{2 \cot(\delta)}{\pi - 2 \delta},$
and ignoring the tails of the distributions of $K$ and $S$, we obtain:
$\mathbb{E}\{\text{ACRB}\} = \int_{0}^{N + 5\sqrt{N}} \int_{0}^{10N} \frac{f_S(s) f_K(k)\frac{6\sigma_R^2}{L|c_3|^2\pi^2MN}}{\gamma_1|\alpha|^2(M^2-1)+\gamma_2(N^2-1)(1-\frac{1}{K-\frac{S}{N}}))}\, dK \, dS$.
\section{Proof of Lemma \ref{lemmapdfuser}}\label{proofpdfuser}
Based on (\ref{x}), and given the signal received at the user as described in Section \ref{systemmodel}, along with the orthogonality of \(\mathbf{v}_i\) and \(\mathbf{s}_u\), the SINR at the user is given by: 
$\text{SINR}_u=\frac{P\tau|c_1|^2|\mathbf{h}^H\mathbf{t}_1|^2}{\sigma^2+\frac{|c_1|^2P(1-\tau)\sum_{i=3}^{N}|\mathbf{h}^H\mathbf{t}_i|^2}{N-2}}$. Due to the orthogonality between \(\mathbf{h}\) and \(\mathbf{t}_i\), this expression simplifies to:  
$\text{SINR}_u=\frac{P\tau|c_1|^2}{\sigma^2_u}|\mathbf{h}^H\mathbf{t}_1|^2.
$ 
\textbf{Upper Bound1:} $|\mathbf{h}^H \mathbf{t}_1|^2 \leq \|\mathbf{h}\|^2,$ so $\text{SINR}_u \leq \frac{P \tau |c_1|^2}{\sigma_u^2} \|\mathbf{h}\|^2.$ Since \(\|\mathbf{h}\|^2 \sim \text{Gamma}(N, 1)= \frac{x^{N-1} e^{-x}}{(N-1)!}, \quad x \geq 0\), an upper bound for $E_{\mathbf{h,a}} \left[ \log(1 + \text{SINR}_u) \right]$ is:
\begin{align}
&\mathbb{E}[\log(1 + \frac{P \tau |c_1|^2}{\sigma_u^2} \|\mathbf{h}\|^2)]\!\!=\!\!\!\!\int_0^\infty \!\!\!\! \!\!\!\!\log(1 + \!\!\frac{P \tau |c_1|^2}{\sigma_u^2} x) \cdot \frac{x^{N-1} e^{-x}}{(N-1)!} \, dx\nonumber\\
&=\frac{e^{\sigma_u^2/P \tau |c_1|^2}}{(N-1)!} \cdot G_{2,3}^{3,1}\left(a \middle| \begin{array}{c} 
0, 0 \\ 
0, -1, N-1 
\end{array} \right),
\end{align}
where \( G \) is the Meijer G-function. For arbitrary \( N \), using Taylor expansion around \( \mathbb{E}[|h|^2] = N \) we have: $\mathbb{E}[\log(1 + \frac{P \tau |c_1|^2}{\sigma_u^2}|h|^2)] \approx \log(1 + \frac{P \tau |c_1|^2}{\sigma_u^2} N) - \frac{a^2}{(1 + \frac{P \tau |c_1|^2}{\sigma_u^2} N)^2} \cdot \frac{N}{2}$.

\textbf{Upper Bound 2:} We have $|\mathbf{h}^H \mathbf{t}_1|^2 = |\alpha|^2 |\tilde{\mathbf{a}}^H \mathbf{h}|^2 + |\beta|^2 (\|\mathbf{h}\|^2 - |\tilde{\mathbf{a}}^H \mathbf{h}|^2) + 2 \text{Re}\{\alpha^* \beta (\tilde{\mathbf{a}}^H \mathbf{h}) \sqrt{\|\mathbf{h}\|^2 - |\tilde{\mathbf{a}}^H \mathbf{h}|^2}\}.$ Moreover, \(|\tilde{\mathbf{a}}^H \mathbf{h}|^2 \sim \text{Exp}(1)\), so \(\mathbb{E}[|\tilde{\mathbf{a}}^H \mathbf{h}|^2] = 1\).Since \(\|\mathbf{h}\|^2 \sim \text{Gamma}(N, 1)\), so \(\mathbb{E}[\|\mathbf{h}\|^2] = N\). Furthermore, the expectation of $\mathbb{E}[\text{Re}\{\alpha^* \beta (\tilde{\mathbf{a}}^H \mathbf{h}) \sqrt{\|\mathbf{h}\|^2 - |\tilde{\mathbf{a}}^H \mathbf{h}|^2}\}]$ is zero, due to the following geometric intuition. We express \(\mathbf{h}\) in a basis where \(\tilde{\mathbf{a}}\) is the first basis vector. Let: $\mathbf{h} = (\tilde{\mathbf{a}}^H \mathbf{h}) \tilde{\mathbf{a}} + \mathbf{h}_\perp,$ where \(\mathbf{h}_\perp\) is the component orthogonal to \(\tilde{\mathbf{a}}\). Then: $\|\mathbf{h}\|^2 = |\tilde{\mathbf{a}}^H \mathbf{h}|^2 + \|\mathbf{h}_\perp\|^2.$
Thus: $\mathbb{E}[\text{Re}\{\alpha^* \beta (\tilde{\mathbf{a}}^H \mathbf{h}) \sqrt{\|\mathbf{h}\|^2 - |\tilde{\mathbf{a}}^H \mathbf{h}|^2}\}] = \mathbb{E}[\text{Re}\{\alpha^* \beta (\tilde{\mathbf{a}}^H \mathbf{h})\}|\mathbf{h}_\perp\|] =0.$ The equality is due to independency of \(\mathbf{h}_\perp\) and \(\tilde{\mathbf{a}}^H \mathbf{h}\), and \(\tilde{\mathbf{a}}^H \mathbf{h}\) is symmetric in phase. Finally, we have: $\mathbb{E}[\text{SINR}_u] = \frac{P \tau |c_1|^2}{\sigma_u^2} (|\alpha|^2 + |\beta|^2 (N - 1)).$ thus, using Jensen's inequality, we have: $\mathbb{E}[\log(1 + \text{SINR}_u)] \leq \log\left(1 + \mathbb{E}[\text{SINR}_u]\right)= \log\left(1 + \frac{P \tau |c_1|^2}{\sigma_u^2} (|\alpha|^2 + |\beta|^2 (N - 1))\right).$
\section{Proof of Lemma \ref{lemmapdfeav}}\label{proofpdfeav}
Based on (\ref{x}), and given the received signal at the communication eav as described in Section \ref{systemmodel}, along with the orthogonality of \(\mathbf{v}_i\) and \(\mathbf{s}_u\), the SINR at the communication eav is given by:  
\begin{align}
\text{SINR}_e\!\!&\overset{(a)}{=}\frac{P\tau|c_2|^2|\mathbf{h}^H_e\mathbf{t}_1|^2}{\sigma^2+\frac{|c_2|^2P(1-\tau)\sum_{i=3}^{N}|\mathbf{h}^H_e\mathbf{t}_i|^2}{N-2}},\label{sinr2i}
\end{align} 
We note that the matrix \(\begin{bmatrix} \mathbf{\tilde{a}} & \mathbf{\tilde{h}} & \mathbf{G} \end{bmatrix}\), where \(\mathbf{G}=\begin{bmatrix} \mathbf{t}_3 & \dots & \mathbf{t}_N \end{bmatrix}\), is unitary, and the entries of \(\mathbf{h}_e\) are i.i.d. zero-mean circularly symmetric complex normal random variables. Since the generation of \(\begin{bmatrix} \mathbf{\tilde{a}} & \mathbf{\tilde{h}} & \mathbf{G} \end{bmatrix}\) is entirely determined by the realizations of \(\mathbf{h}\) and \(\mathbf{a}\), which are independent of \(\mathbf{h}_e\) (as stated in Section \ref{systemmodel}), we conclude that \(\mathbf{h}_e\) and \(\begin{bmatrix} \mathbf{\tilde{a}} & \mathbf{\tilde{h}} & \mathbf{G} \end{bmatrix}\) are mutually independent.  
Thus, the distribution of \(\mathbf{h}^H_e\begin{bmatrix} \mathbf{\tilde{a}} & \mathbf{\tilde{h}} & \mathbf{G} \end{bmatrix}\) follows the same distribution as \(\mathbf{h}_e\), i.e., $\mathbf{h}^H_e\begin{bmatrix} \mathbf{\tilde{a}} & \mathbf{\tilde{h}} & \mathbf{G} \end{bmatrix} \sim \mathcal{CN}(0,\mathbf{I}).$ As a result, we can express $\mathbf{h}^H_e\mathbf{t}_1 = \alpha\mathbf{h}^H_e\mathbf{\tilde{a}} + \beta \mathbf{h}^H_e\mathbf{\tilde{h}},$ where \(\mathbf{h}^H_e\mathbf{\tilde{a}} \sim \mathcal{CN}(0,1)\) and \(\mathbf{h}^H_e\mathbf{\tilde{h}} \sim \mathcal{CN}(0,1)\). Since this expression is a linear combination of independent complex Gaussian random variables, it follows a complex Gaussian distribution with mean $\alpha E[\mathbf{h}^H_e\mathbf{\tilde{a}}] + \beta E[\mathbf{h}^H_e\mathbf{\tilde{h}}] = 0$ and variance $E[|\alpha\mathbf{h}^H_e\mathbf{\tilde{a}} + \beta \mathbf{h}^H_e\mathbf{\tilde{h}}|^2] = |\alpha|^2 + |\beta|^2 = 1.$ Thus, \(\mathbf{h}^H_e\mathbf{t}_1 \sim \mathcal{CN}(0,1)\), which implies \(|\mathbf{h}^H_e\mathbf{t}_1|^2 \sim \chi^2_2\).  
Furthermore, for \(i \in \{3, \dots, N\}\), we have \(\mathbf{h}^H_e\mathbf{t}_i \sim \mathcal{CN}(0,1)\), which are i.i.d. random variables. Therefore, the sum \(\sum_{i=3}^{N} |\mathbf{h}^H_e\mathbf{t}_i|^2\) follows a Chi-Square distribution with \(2(N-2)\) degrees of freedom, i.e., $\sum_{i=3}^{N} |\mathbf{h}^H_e\mathbf{t}_i|^2 \sim \chi^2_{2(N-2)}.
$ Moreover, since $\mathbf{h}^H_e\begin{bmatrix} \mathbf{\tilde{a}} & \mathbf{\tilde{h}} & \mathbf{t}_3 & \dots & \mathbf{t}_N \end{bmatrix} \sim \mathcal{CN}(0,\mathbf{I}),$ the random variables \(|\mathbf{h}^H_e\mathbf{t}_1|^2\) and \(\sum_{i=3}^{N} |\mathbf{h}^H_e\mathbf{t}_i|^2\) are independent. By defining \( C_1 = \frac{P \tau |c_2|^2}{\sigma^2}, \quad C_2 = \frac{P|c_2|^2(1-\tau)/(N-2)}{\sigma^2} \), \( X \triangleq |\mathbf{h}^H_e \mathbf{t}_1|^2 \sim \chi^2_2 \), \( Y \triangleq \sum_{i=3}^{N} |\mathbf{h}^H_e \mathbf{t}_i|^2 \sim \chi^2_{2(N-2)} \), where \( X \) and \( Y \) are independent, we have \( \text{SINR}_e = \frac{C_1 X}{1 + C_2 Y} \). Thus:
\begin{align}
E_{X,Y} \left[ \log(1 + \text{SINR}_e) \right]&\overset{(a)}{=}\int_{0}^{\infty}\!\!\!\!P\left(X > \frac{T + TC_2Y}{C_1}\right)dt\label{simplifiedeav1}
\end{align}
where \( T = 2^t - 1 \), and (a) is due to the statistical property. Since \( X \sim \text{Exp}(1/2) \), its complementary CDF (tail probability) is: \( P(X > x) = e^{-x/2}, \quad x \geq 0 \). Thus, we get: \( P\left( X > \frac{T + TC_2Y}{C_1} | Y \right) = \exp\left( - \frac{T + TC_2Y}{2C_1} \right) \). Therefore, we have:
\begin{align}
&P\left(X > \frac{T + TC_2Y}{C_1}\right)\overset{(a)}{=}\mathbb{E}_Y\left[\exp\left(-\frac{T + TC_2Y}{2C_1}\right)\right]\nonumber\\
&=\exp\left(-\frac{T}{2C_1}\right) \cdot \mathbb{E}_Y\left[\exp\left(-\frac{TC_2Y}{2C_1}\right)\right]\nonumber\\
&\overset{(b)}{=}\exp\left(-\frac{T}{2C_1}\right)\left(1 + \frac{TC_2}{C_1}\right)^{-(N-2)}\label{simplifiedeav2}
\end{align}
(a) follows from statistical properties, and (b) holds because the term \( \mathbb{E}_Y\left[ \exp\left( - \frac{TC_2Y}{2C_1} \right) \right] \) is the moment-generating function (MGF) of \( Y \), evaluated at \( s = - \frac{TC_2}{2C_1} \). For \( Y \sim \text{Gamma}(N-2, 2) \), the MGF is given by \( M_Y(s) = \mathbb{E}_Y\left[ e^{sY} \right] = \left( 1 - \theta s \right)^{-k} \), where \( \theta = 2 \) is the scale parameter and \( k = N-2 \) is the shape parameter. Thus, replacing (\ref{simplifiedeav2}) into (\ref{simplifiedeav1}), the proof is complete.
\end{document}